\documentclass[conference]{IEEEtran}
\IEEEoverridecommandlockouts
\usepackage{cite}
\usepackage{amsmath,amssymb,amsfonts}
\usepackage{graphicx}
\usepackage{textcomp}
\usepackage{xcolor}
\usepackage{balance}  % for  \balance command ON LAST PAGE  (only there!)

\usepackage{caption}

\usepackage{booktabs}
\usepackage{verbatim}
\usepackage{graphicx}
\usepackage{balance}
\usepackage{amsmath}
\usepackage{verbatim}
\usepackage[plain]{algorithm}
\usepackage{algpseudocode}
\usepackage{pgfplots}
\usepackage{pgfplotstable}
\usepgfplotslibrary{groupplots}
\usepackage[center]{subfigure}
\usepackage{url}
\usetikzlibrary{patterns,shapes,arrows,positioning,fit,decorations.pathreplacing}
\usepackage{xifthen}
\usepackage{xcolor}
\usepackage{cite}
\usepackage{pgf-pie}

\usepackage{amsthm}

\newtheorem{lemma}{Lemma}
\newtheorem{theorem}{Theorem}
\newtheorem{corollary}{Corollary}

\newtheorem{definition}{Definition}

\theoremstyle{definition}
\newtheorem{example}{Example}

\theoremstyle{plain}

\pgfplotsset{compat=1.13}

\def\BibTeX{{\rm B\kern-.05em{\sc i\kern-.025em b}\kern-.08em
    T\kern-.1667em\lower.7ex\hbox{E}\kern-.125emX}}
\begin{document}

%\newpage

%\include{replies}

%\title{Efficiently Generating Voice Data Summaries \\via Pre-Generated Speech Fragments}
%\title{Answering Voice Queries Efficiently\\via Pre-Generated Data Summaries}
%\title{Selecting Optimal Fact Combinations \\for Summarizing Voice Query Results}

\title{Optimally Summarizing Data by Small Fact Sets\\for Concise Answers to Voice Queries}

\author{\IEEEauthorblockN{1\textsuperscript{st}Immanuel Trummer}
\IEEEauthorblockA{\textit{Database Group} \\
\textit{Cornell University}\\
Ithaca, NY \\
it224@cornell.edu}
\and
\IEEEauthorblockN{2\textsuperscript{nd} Connor Anderson}
\IEEEauthorblockA{\textit{Database Group} \\
\textit{Cornell University}\\
Ithaca, NY \\
ca339@cornell.edu}
}

\maketitle

\begin{abstract}
Our goal is to find combinations of facts that optimally summarize data sets. We consider this problem in the context of voice query interfaces for simple, exploratory data analysis. Here, the system answers voice queries with a short summary of relevant data. Finding optimal voice data summaries is computationally expensive. Prior work in this domain has exploited sampling and incremental processing. Instead, we rely on a pre-processing stage generating summaries of data subsets in a batch operation. This step reduces run time overheads by orders of magnitude. 

We present multiple algorithms for the pre-processing stage, realizing different tradeoffs between optimality and data processing overheads. We analyze our algorithms formally and compare them experimentally with prior methods for generating voice data summaries. We report on multiple user studies with a prototype system implementing our approach. Furthermore, we report on insights gained from a public deployment of our system on the Google Assistant Platform.
\end{abstract}


\section{Introduction}

Our goal is to optimally summarize data with a bounded number of facts. This problem arises in the context of voice query interfaces~\cite{Lyons2016, Saul2019, Yang2019, Joglekar2015a, Trummer2019a} (VQIs). Motivated by the rise of devices such as Google Home or Amazon Alexa, VQIs translate speech input to SQL queries and report the query result via voice output. Voice output must be concise~\cite{MILLER1956, Trummer2018}, preventing VQIs from transmitting all but the smallest results in detail. Hence, we focus on approximating data as closely as possible, using a length-limited speech summary. We present multiple algorithms to solve this problem efficiently, exploiting pre-processing to reduce run time overheads.

Our optimality criterion is based on the user's belief about data. We aim to alter the user's expectations via speech output to bring it as close as possible to the actual data. In that, we follow prior work on data vocalization~\cite{Trummer2018a, Trummer2019a} and visualization (where the goal is to select a plot to reduce deviation between user expectations and data~\cite{Sarawagi2000, Sarawagi2000a}).

%Given a \textcolor{red}{voice query}, our goal is to summarize the data using a bounded number of facts. Our search space are fact sets. \textcolor{red}{Similar to prior work on data visualization and vocalization~\cite{Sarawagi2000, Trummer2018a}}, our goal is to \textcolor{red}{present facts that} 

%give users the most precise impression of the data under the aforementioned constraints. We base our optimization goal on a user behavior model that we validate in multiple user studies. %\textcolor{red}{The following example illustrates the problem we address.}

%Intuitively, Speech~2 approximates the data better (i.e., the difference between approximation and data, summed over all fields, is smaller for Speech~2). We formalize this intuition in the following section. 

The resulting problem is NP-hard (see Section~\ref{analysisSec}) and challenging to solve. Prior work resorted to sampling~\cite{Trummer2019a, Trummer2018, Trummer2018a} to generate approximate summaries at run time. We present a new approach that is based on the following hypothesis: typical voice queries are short and simple (as it becomes tedious to formulate long queries without being able to see and edit them). Hence, we can generate answers for all queries up to a certain (small) length in an efficient batch operation. In the following, we present exact and approximate algorithms for this pre-processing step, evaluate their performance, and assess them in user studies.

%In this work, we explore the possibility of pre-processing data subsets to speed up voice summarization at run time.

%We focus particularly on scenarios where voice query backends run in the Cloud. This scenario is realistic when voice applications are served via frameworks such as Google Assistant (which is the case for our prototype). Our primary metrics are then latency as well as the processing costs per voice query, payable to the Cloud provider. We show that our approach reduces monetary processing costs by orders of magnitude, compared to another recent baseline generating voice summaries. 

In summary, our original scientific contributions are threefold. First, we propose several novel algorithms that generate optimized speech summaries of data subsets. Second, we analyze those algorithms formally. Third, we analyze the algorithms experimentally, in terms of computational efficiency and by a user study.

\begin{comment}
\begin{itemize}
\item We are the first to explore pre-processing as a means of avoiding run time overheads when serving user groups via voice query interfaces.
\item We propose multiple approaches to generate speech summaries in a pre-processing step, including exhaustive search, greedy algorithms, and different methods to improve efficiency by pruning.
\item We evaluate the proposed approaches experimentally, including multiple user studies in which we compare voice against visual interfaces.
\end{itemize}
\end{comment}

The reminder of this paper is organized as follows. In Section~\ref{modelSec}, we introduce our formal problem model. In Section~\ref{overviewSec}, we give an overview of the end-to-end voice querying engine in which the proposed methods are used. Next, we describe an exhaustive algorithm for that problem in Section~\ref{exactSec}. Then, we present a greedy algorithm to solve the problem in Section~\ref{greedySec}. We show that this algorithm generates guaranteed near-optimal speech summaries. Next, we show in Section~\ref{pruningSec} how to improve efficiency by pruning irrelevant speech fragments early. We analyze complexity of the target problem and of the proposed algorithms in Section~\ref{analysisSec}. In Section~\ref{experimentsSec}, we evaluate all proposed approaches experimentally. Finally, we discuss related work in Section~\ref{relatedSec} before we conclude.

%Also, we report on results of multiple user studies. 

%This paper focuses on a pre-processing step in which speech fragments are generated. We introduce the corresponding problem model in Section~\ref{modelSec}. 

%for the NYC area from the US Census American Community Survey health-related statistics by the US Census popular data set from the NYC Open Data domain, capturing 311 service requests\footnote{\url{https://data.cityofnewyork.us/Social-Services/311-Service-Requests-from-2010-to-Present/erm2-nwe9}}. 

\section{Problem Model}
\label{modelSec}

\tikzstyle{ontime}=[fill=gray!10]
\tikzstyle{slight}=[fill=yellow!50, draw=none]
\tikzstyle{moderate}=[fill=orange!50, draw=none]
\tikzstyle{late}=[fill=red, draw=none]

\begin{figure}[t]
    \centering
    \begin{tikzpicture}
    \begin{groupplot}[group style={group size=3 by 1, y descriptions at=edge left}, width=3cm, height=3cm, view={0}{90}, xticklabels={East, South, West, North}, xticklabel style={rotate=60}, yticklabels={Spring, Summer, Fall, Winter}, xmin=0, xmax=4, ymin=0, ymax=4, xtick={0.5, 1.5, 2.5, 3.5}, ytick={0.5, 1.5, 2.5, 3.5}, point meta min=0, point meta max=20]
        \nextgroupplot[title={Data}]
        \draw[ontime] (axis cs: 0,0) rectangle (axis cs:4,4);
        \draw[late] (axis cs: 0,3) rectangle (axis cs:1,4);
        \draw[slight] (axis cs: 1,3) rectangle (axis cs:4,4);
        \draw[late] (axis cs: 3,3) rectangle (axis cs:4,4);
        \draw[slight] (axis cs: 3,0) rectangle (axis cs:4,1);
        \draw[slight] (axis cs: 3,2) rectangle (axis cs:4,3);
        \draw[late] (axis cs: 3,1) rectangle (axis cs:4,2);
        \draw[late] (axis cs: 1,1) rectangle (axis cs:2,2);
        \draw[black] (axis cs: 0,0) rectangle (axis cs:4,4);
        \nextgroupplot[title={Speech 1}]
        \draw[ontime] (axis cs: 0,0) rectangle (axis cs:4,4);
        \draw[late] (axis cs: 1,1) rectangle (axis cs:2,2);
        \draw[late] (axis cs: 0,3) rectangle (axis cs:1,4);
        \draw[black] (axis cs: 0,0) rectangle (axis cs:4,4);
        \nextgroupplot[title={Speech 2}, area legend, legend to name=delayLegend, legend columns=4]
        \draw[ontime] (axis cs: 0,0) rectangle (axis cs:4,4);
        \draw[moderate] (axis cs: 0,3) rectangle (axis cs:4,4);
        \draw[moderate] (axis cs: 3,0) rectangle (axis cs:4,4);
        \draw[black] (axis cs: 0,0) rectangle (axis cs:4,4);
        \addlegendimage{fill=gray!10}
        \addlegendentry{No Delay}
        \addlegendimage{fill=yellow!50}
        \addlegendentry{10 Mins}
        \addlegendimage{fill=orange!50}
        \addlegendentry{15 Mins}
        \addlegendimage{fill=red}
        \addlegendentry{20 Mins}
    \end{groupplot}
    \end{tikzpicture}
    
    \ref{delayLegend}
    \caption{Data on airplane delays as a function of region and season (left plot) and approximation created by two alternative speeches (middle and right plot).}
    \label{fig:delays}
\end{figure}
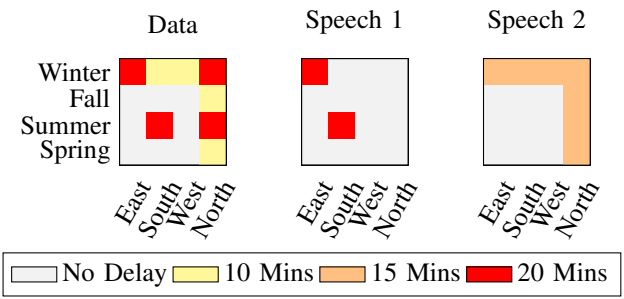

We introduce our problem model using examples.

%Our goal is to summarize relational data via concise speech summaries. We introduce a corresponding model and terminology. In the following, we assume that summarization focuses on how values in a target column depend on values in dimension columns. For the same database relation, we can create multiple summarization problem instances that use different target and dimension columns.

\begin{definition}
We model a \textbf{Relation} as a set $R$ of rows. Each \textbf{Row} $r=\langle D_r,v_r\rangle\in R$ is characterized by a pair where $D_r=\{\langle d_i,v_i\rangle\}$ assigns \textbf{Dimension Columns} $d_i$ to values $v_i$ while $v_r$ is a numerical value in the \textbf{Target Column} for $r$.
\end{definition}

\begin{example}
Consider a relation where each row represents delay of a flight. We are interested in how flight delays depend on region and season. Hence, the column containing delays (e.g., in minutes) is our target, we use columns containing season and region as dimensions. Figure~\ref{fig:delays} illustrates average delays for different data subset in the left plot. We will use this scenario as running example in the following.
\end{example}

\begin{comment}
\begin{example}
\textcolor{red}{Consider data on average flight delays for different regions and seasons, illustrated in Figure~\ref{fig:delays} (left). Our goal is to approximate the data by a speech containing at most two simple facts (i.e., averages for data subsets defined by a single region or seasons or both). For instance, we could convey average delays for the South in summer and for the East in Winter (Speech~1). Alternatively, we could convey averages for Winter and for the North (Speech~2). Figure~\ref{fig:delays} illustrates the approximation conveyed by those speeches. Speech~2 minimizes the distance between the approximation and actual data. Hence, it is the better summary. Our goal is to efficiently find optimal speech summaries.}
\end{example}
\end{comment}

For the following definitions, we assume that a relation has been fixed. We consider a simple class of facts that can be easily described via speech output.

\begin{definition}
A \textbf{Fact} $f=\langle D,v\rangle$ is a pair, consisting of a \textbf{Scope} $D$ and a \textbf{Typical Value} $v$. Scope $D=\{\langle d_i,v_i\rangle\}$ assigns values to a subset of dimension columns while the typical value $v$ averages over the target column. We say that a row $r=\langle D_r,v_r\rangle\in R$ is \textbf{Within Scope} for the fact if $D\subseteq D_r$ (i.e., fact and row values are consistent in the dimension columns). The typical value is the average value in the target column for all rows within scope.
\end{definition}

%\textcolor{red}{We illustrate those concepts by a running example.}

We assume that a speech template is given that allows translating facts of the form above into speech output. This template contains placeholders for the typical value and for (a variable number of) dimension columns.

\begin{definition}
For our pseudo-code, we model a \textbf{Speech} $F$ as a set of facts. We call its cardinality the \textbf{Speech Length}.
\end{definition}

\begin{example}
The following facts can be derived from the data depicted in Figure~\ref{fig:delays}, assuming the same number of flights for each season and region. The average delay in Summer in the South is 20 minutes (i.e., $F=\langle D,20\rangle$ and $D=\{\langle season,Summer\rangle,\langle region,South\rangle\}$). The average delay in Winter is 15 minutes ($D=\{\langle season,Winter\rangle\}$). A speech corresponds to a set of such facts.
\end{example}

%here, the scope is defined by restricting only one dimension

%the scope is defined by restrictions on both dimension columns

%Listening to a speech \textcolor{red}{creates} expectations about the summarized data. We model user expectations via a simple model.

We need a criterion to select between alternative fact combinations. We follow prior work on data visualization and vocalization~\cite{Sarawagi2000, Trummer2018a} that aims at selecting facts to bring the user's expectations as close as possible to the actual data. Next, we introduce a user model, modelling how user expectations change after listening to speeches.

\begin{definition}
We denote by $E(F,r)$ the \textbf{Expected Value} in the target column of row $r=\langle D_r,v_r\rangle$ after listening to facts $F$. Denote by $F_r\subseteq F$ the subset of facts for which $r$ is within scope. If $|F_r|=1$ then only one fact is relevant to the row and user expectation equals the typical value proposed by that fact. If $|F_r|>1$ then multiple facts are relevant. Denote by $V_r$ the set of typical values proposed by any fact in $F_r$. We assume that users often have prior knowledge allowing them to determine the most relevant fact among alternatives. Hence, we model user expectations as $\arg\min_{v\in V_r}|v-v_r|$ as the value closest to the target. Also, we consider \textbf{Prior Expectations} by the user (before listening) via function $P(r)$. The prior value is included in the set $V_r$ for any row.
\end{definition}

In Section~\ref{experimentsSec}, we match this model against expectations of actual users in a user study.

\begin{example}
Assume users expect no delays by default (the prior). After a speech conveying average delays in the South in Summer and in the East in Winter, the resulting expectations by users for flights in specific regions and seasons are depicted in Figure~\ref{fig:delays} (middle). If conveying average delays in Winter and average delays in the North (i.e., two facts), the resulting expectations are illustrated in the rightmost plot.
\end{example}

%The goal of speech output is to give users information on data. We measure the amount of information gained via a specific speech, its utility, by calculating deviation between user's expectations and data, before and after listening to the speech. An informative speech reduces deviations significantly.

%{The goal of querying is to gain information on data. Hence,} a speech is useful if it reduces deviation \textcolor{red}{between expectations of the listeners and the actual data}.

%Deviation is the distance between expected and actual values.

\begin{definition}
We denote by $D(F,r)$ for facts $F$ and row $r=\langle D_r,v_r\rangle$ the \textbf{Deviation} between user expectations and actual value, i.e.\ $D(F,r)=|E(F,r)-v_r|$. We denote by $D(F)=\sum_{r\in R}D(F,r)$ the accumulated deviation (also ``\textbf{Error}'') over all rows of the relation. 
\end{definition}

\begin{definition}
We denote by $U(F)$ the \textbf{Utility} of facts $F$ in bringing the user's expectations closer to actual values. More precisely, it is $U(F)=D(\emptyset)-D(F)$. This definition relies on having a prior defining user expectations in the absence of relevant facts. We use the term \textbf{Single-Fact Utility} for the utility of a singleton fact set.
\end{definition}

%During the following sections, we occasionally distinguish between single-fact and aggregate utility. Aggregate utility refers to the utility of a speech, single-fact utility to the utility of a single fact within that speech alone. Using speech output, we can only transmit a limited number of facts~\cite{MILLER1956, Trummer2018}. 

\begin{example}\label{ex:utility}
Consider expectations induced by Speech~1 and Speech~2 of our running example (see Figure~\ref{fig:delays}). As outlined before, users expect no delays by default, leading to an accumulated error of $4\cdot 20+4\cdot 10=120$ (assuming one relation row per combination of season and region). After listening to Speech~1, error reduces to $80$ (the delta of $40$ is the utility of Speech~1). On the other hand, Speech~2 reduces error to $7\cdot 5=35$ and is therefore more useful.
\end{example}

Our goal is to find the combinations of facts with limited cardinality that is most useful.

\begin{definition}\label{def:summarization}
An instance of the \textbf{Speech Summarization Problem} is defined by a triple $\langle R,F,m\rangle$ where $R$ is a relation to summarize, $F$ a set of available facts, and $m$ the maximal number of facts to use. The goal is to find up to $m$ facts $F^*$ for speech output maximizing utility, i.e.\ $\arg\max_{F^*\subseteq F:|F^*|=m}U(F^*)$. 
\end{definition}

\begin{comment}
\begin{example}
Consider the two speeches in Table~\ref{twoSpeechesTable} (Section~\ref{experimentsSec}), summarizing the data shown in Figure~\ref{fig:visp} (label ``Correct''). Both contain three facts, each fact specifying average target values (here: prevalence of visual impairment) for data subsets (here: subsets defined by location and age group). The second speech has higher utility according to our model as it induces an expectation closer to the actual data. This can be seen in Figure~\ref{fig:visp}, comparing predictions of users after listening to different speeches to the correct value.
\end{example}
\end{comment}

\section{System Overview}
\label{overviewSec}

\tikzstyle{label}=[align=center, font=\itshape]
\tikzstyle{component}=[fill=blue!20, font=\itshape\bfseries, minimum width=3.25cm, draw, align=center]
\tikzstyle{flow}=[ultra thick, ->]

\begin{figure}[t]
    \centering
    \begin{tikzpicture}
        \draw[fill=red!10] (-4,0) rectangle (0,2.5);
        \draw[fill=green!10] (0,0) rectangle (4,2.5);
        
        \node[label] (config) at (-2,2.75) {Configuration};
        \node[component] (generator) at (-2,2) {Problem Generator};
        \node[component, fill=red] (summarizer) at (-2,1.25) {Speech Summarizer};
        \node[component] (dbms) at (-2,0.5) {Relational DBMS};
        \draw[flow] (config) -- (generator);
        \draw[flow] (generator) -- (summarizer);
        \draw[flow, <->] (summarizer) -- (dbms);
        
        \node[label] (query) at (2,2.75) {Voice Query};
        \node[component] (stt) at (2,2) {Speech Recognition};
        \node[component] (text) at (2,1.25) {Text to Query};
        \node[component] (qts) at (2,0.5) {Query to Speech};
        \draw[flow] (query) -- (stt);
        \draw[flow] (stt) -- (text);
        \draw[flow] (text) -- (qts);
        
        \node[label] (speeches) at (0,2.75) {Speeches};
        \draw[flow] (dbms.east) -- (speeches);
        \draw[flow] (speeches) -- (qts.west);
        \node[label] (output) at (4,2.75) {Voice Output};
        \draw[flow] (qts.east) -- (output);
    \end{tikzpicture}
    \caption{Illustration of voice querying system.}
    \label{overviewFig}
\end{figure}
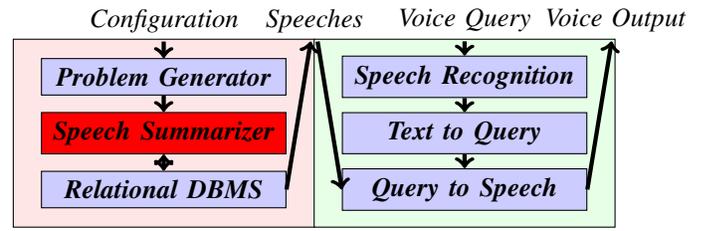

The algorithms presented in this paper are executed in the Speech Summarizer component of the system illustrated in Figure~\ref{overviewFig}. This system answers simple voice queries via voice output, based on the speeches generated during pre-processing. A video, showing a recent demonstration~\cite{Trummer2020}, is available\footnote{\url{https://youtu.be/awUtgXB60Bk}}.

The goal of pre-processing is to pro-actively generate speech answers for all possible queries, up to a certain length. This may not seem practical for written queries. However, we hypothesized that voice input would motivate users to formulate relatively short and simple queries, making pre-processing a viable option. We found that hypothesis to be largely justified (e.g., see Figure~\ref{fig:queryComplexity}). The queries to consider are described in a Configuration file. Based on that, the Problem Generator generates a series of speech summarization problems. Those problems are solved by the Speech Summarizer, using a relational database engine and the algorithms described later.

%The system currently supports a simple model, letting users select a column and a subset of rows via equality predicates (the answer of the system may however use more columns than appear in the query). 

Currently, we consider queries requesting information on values in a target column for a data subset, defined by a conjunction of equality predicates (the answer of the system may reference columns that do not appear in the query). We found this to be the most common type of voice query in user studies (see Section~\ref{sub:deployment}). Query length is measured by the number of equality predicates. The Configuration file references a table in a relational database. It specifies the maximal query length to consider, the columns on which to allow predicates (we call them ``Dimensions''), and a set of target columns. The Problem Generator creates one query for each combination of a target column and a subset of equality predicates, considering all possible combinations of equality predicates up to the query length. For each such query, we generate a speech summarizing values in the target column for the data subset defined by the query predicates. The facts considered for summarization report average values in the target column for data subsets. We consider one fact for each data subset defined by a conjunction of the query predicates and, by default, up to two additional equality predicates on the dimensions (considering equality predicates for all value combinations that appear in the data set). After selecting a (near-)optimal fact combination, our research focus, the speech is generated according to a simple text template. The summarization methods, presented in this paper, are not specific to the type of queries currently supported, nor specific to our current method of selecting facts.

At run time, the system maps voice queries to the most related speech summary, generated during pre-processing. Each voice query is mapped to one target column and a set of equality predicates. To map text to queries, we train an extractor with a few samples to extract names of target column and predicates on other columns (we currently consider equality predicates only) from input text (this functionality is provided by the Google Assistant framework on which our application is based upon\footnote{\url{https://developers.google.com/assistant}}). If a summary was generated for the extracted target column and for the data subset defined by the extracted predicates, the corresponding speech is vocalized. Otherwise, among all speeches referencing the queried target column, the speech describing the most specific data subset that contains the one referenced in the query is used. More precisely, considering predicates $Q$ extracted from the query, we select a speech summarizing a data subset defined by predicates $S$ such that $S\subseteq Q$ and $|S\cap Q|$ is maximal. Speeches are prefixed with a description of the summarized data subset (an enumeration of restricted columns and their values). Hence, users are aware of the semantics of the description.

\begin{example}
We made a data set on flight cancellations publicly available via voice queries (see Section~\ref{sub:deployment}). Here, we specified one target column (cancellation probability) and six dimensions (e.g., airline and start airport state), enabling queries with up to two predicates. Based on that configuration, we generated 8,500 speeches during pre-processing. In the logs, we found for instance the voice query ``cancellations in Winter?''. After receiving this query, the system will extract the target column (cancellation) and the predicate (season = Winter) of the query, then search for matching, pre-generated speeches. Here, as the query has one predicate, a speech has been pre-generated for that precise query. This speech is transmitted via voice output. It describes three facts (the general cancellation probability, a significant increase in February, and a reduced probability for the West) that together provide the best possible approximation of the data.
\end{example}

\section{Exact Algorithm}
\label{exactSec}

The exact algorithm finds guaranteed optimal speech summaries. We give an overview in Section~\ref{sub:exactAlg}, discuss pruning in Section~\ref{exactPruningSub}, and prove optimality in Section~\ref{sub:optimality}.

%We describe an exhaustive optimization algorithm in Section~\ref{sub:exactAlg}. In Section~\ref{sub:optimality}, we prove that it generates optimal speech summaries.

\subsection{Algorithm Overview}
\label{sub:exactAlg}

The exact algorithm is iterative. Starting from single facts, it expands speeches in each iteration until speeches reach maximum length. While doing so, it prunes speeches that provably cannot expand into an optimal speech. For pruning, it exploits a lower bound on the utility of the optimal speech (generated by a cheaper heuristic). Furthermore, the algorithm avoids enumerating redundant permutations of fact sets.

\begin{algorithm}[t]
\captionsetup{labelfont={color=black},font={color=black}}
\renewcommand{\algorithmiccomment}[1]{// #1}
\begin{algorithmic}[1]
\State \Comment{Find guaranteed optimal summary for}
\State \Comment{relation $R$ using $m$ facts from $F$. Exploit }
\State \Comment{lower bound $b$ on optimal speech utility.}
\Function{OptimalSummary}{$F,R,m,b$}
\State \Comment{Calculate utility for single fact speeches}
\State $S\gets\Gamma_{\sum\mathcal{U},F}(R\Join_{\mathcal{M}}F)$
\State \Comment{Iteratively combine facts into speeches}
\For{$i\in2,\ldots,m$}
\State \Comment{Expand speech and prune}
\State $S\gets\sigma_{\mathcal{P}(b,m-i-1)}(\Pi_{\mathcal{\widetilde{U}},S,F}(S\times F))$
\EndFor
\State \Comment{Calculate utility of final speeches}
\State $S\gets\Gamma_{\sum\mathcal{U},S}(R\Join_{\mathcal{M}}S)$
\State \Comment{Return speech with maximal utility}
\State \Return{$\arg\max_{\mathcal{U}}(S)$}
\EndFunction
\end{algorithmic}
\caption{Exhaustive algorithm for generating guaranteed optimal speech summaries.\label{exactAlg}}
\end{algorithm}

Algorithm~\ref{exactAlg} shows the associated pseudo-code. Given facts $F$ and relation $R$ as input, as well as the maximal speech length $m$ and a lower utility bound $b$, the algorithm returns an optimal combination of up to $m$ facts to summarize $R$.

The algorithm is executed as a series of relational operators ($\Gamma$ for grouping and aggregation, $\sigma$ for filtering, $\Pi$ for projection, $\Join$ for joins, and $\times$ for the Cartesian product). Our implementation executes the algorithm by issuing a series of SQL queries (thereby removing the need for transferring data out of the database system). Initially, Algorithm~\ref{exactAlg} calculates the associated utility for each fact (Line~6). Utility of a fact (or speech) is calculated by summing up utility over each row (hence, we use sum aggregation while grouping by facts). We require a join as each fact must be matched against all rows that fall within its scope. The corresponding join condition ($\mathcal{M}$) compares values in $F$ and $R$ for each dimension $d$, requiring $F.d=\mathbf{null}$ or $F.d=R.d$ for each of them. In each iteration, speeches in $S$ are expanded by adding one more fact (considering all possible facts). Furthermore, speeches are pruned via condition $\mathcal{P}$, based on several utility-related bounds (represented as~$\mathcal{\widetilde{U}}$ in Algorithm~\ref{exactAlg}). Pruning is described in the next subsection. Next, the algorithm calculates the precise utility of each remaining speech by a join between data and speeches (Line~13). Here, $\mathcal{M}$ evaluates to true if the data row falls within the scope of at least one speech fact. Finally, the algorithm returns the speech with maximal utility (Line~15).

\subsection{Pruning Speeches}
\label{exactPruningSub}

First, our utility model does not depend on the order in which facts appear in a speech. We prune out redundant fact permutations by enforcing specific order between facts. Specifically, we order facts in decreasing order of single-fact utility via pruning condition $S.\mathcal{U}_P\geq F.\mathcal{U}$. Here, $S.\mathcal{U}_P$ is the single-fact utility of the last, previously added fact in a speech to expand. $F.\mathcal{U}$ is the single-fact utility of a newly added fact.

\begin{comment}
distinguish based on fact order. Hence, we can reduce the number of speech candidates by considering (ideally) only one single permutation per unique fact set. \textcolor{red}{We only allow speeches that order facts in decreasing order of single-fact utility (i.e., utility of a speech using only that fact). }

We could use various criteria for ordering facts. For reasons justified in the following paragraphs, we specifically choose to order facts in decreasing order of single-fact utility. Denoting by $F.\mathcal{U}$ the single-fact utility of facts considered for expansion, and by $S.\mathcal{U}_P$ the single-fact utility of the previously added fact in current speeches, we use condition

\begin{equation*}
\mathcal{P}_1=S.\mathcal{U}_P\geq F.\mathcal{U}
\end{equation*}

for pruning. This eliminates speech expansions that order higher-utility facts after lower-utility facts. 
\end{comment}

Second, we prune speeches that cannot expand into an optimal speech. We compare an upper bound on the utility of all expansions of a candidate speech to a lower bound on optimal utility (input $b$). To calculate an upper utility bound for all expansions, we consider the remaining number of expansions $r$, the constraint that facts are added in decreasing order of single-fact utility, and that single-fact utility is an upper bound for the increase in utility when adding a fact to a non-empty speech. The corresponding pruning condition, $(b-S.\mathcal{U})/r\leq F.\mathcal{U}$ (where $S.\mathcal{U}$ is the upper utility bound of the speech to expand, obtained by summing utility of its facts, and $F.\mathcal{U}$ the single-fact utility of a candidate expansion), is justified in detail in the next subsection. Condition $\mathcal{P}$, used in Line~10 of Algorithm~\ref{exactAlg}, is the conjunction of the two aforementioned pruning conditions. The second pruning formula depends on $b$ and $r$, specified as parameters to $\mathcal{P}$ in Algorithm~\ref{exactAlg}.

\begin{example}
Reconsider the example depicted in Figure~\ref{fig:delays}. We consider expansions of a speech stating that the average delay in the South in Summer is 20 minutes. Calculating utility under the same assumptions as in Example~\ref{ex:utility}, this fact alone has utility 20. Consider the fact stating that the average delay in Winter is 15 minutes. This fact has single-fact utility 40. We prune the corresponding expansion as it does not order facts by (decreasing) single-fact utility. Now, consider expansion with a fact stating that the average delay in the East in Winter is 20 minutes. Assume we generate speeches with up to two facts. Knowing a speech with utility 85, generated by a heuristic, we can discard the current speech altogether (since $b=85$, $S.\mathcal{U}=20$, $F.\mathcal{U}=20$, $r=1$, and  $(b-S.\mathcal{U})/r> F.\mathcal{U}$).
\end{example}

\subsection{Proof of Optimality}
\label{sub:optimality}

%We prove that Algorithm~\ref{exactAlg} produces optimal speeches. This guarantee is based on a diminishing returns property of the utility function. We define that property next.

First, we introduce a diminishing returns property.

\begin{definition}
A set function $f:S\mapsto\mathbb{R}^+$ is sub-modular if for all sets $S_1,S_2\subseteq S$ where $S_1$ is a subset of $S_2$ ($S_1\subseteq S_2$), the increase by adding a new element $s\in S$ is higher for $S_1$ than for $S_2$: $f(S_1\cup\{s\})-f(S_1)\geq f(S_2\cup\{s\})-f(S_2)$.
\end{definition}

Next, we show that utility, as a function of the set of speech facts, has the aforementioned property.

\begin{theorem}
Speech utility has diminishing returns.\label{th:diminishing}
\end{theorem}
\begin{IEEEproof}
Let $F_1$ and $F_2$ be two speeches (sets of facts) describing the same data such that $F_1\subseteq F_2$. Denote by $f$ a relevant fact that neither appears in $F_1$ nor $F_2$. 
For a fixed row $r$ with dimension values $D_r$ and target value $v_r$, it is $D(F_1,r)=\min_{\langle D,v \rangle\in F_1:D\subseteq D_r}|v-v_r|\geq \min_{\langle D,v \rangle\in F_2:D\subseteq D_r}|v-v_r|=D(F_2,r)$ since $F_1\subseteq F_2$. Denote by $\Delta U(F,f)=D(F,r)-D(F\cup\{f\},r)$ the utility of adding fact $f$. We distinguish two cases: either the new fact $f$ provides a better approximation than any fact in $F_2$ (case 1) or not (case 2). In case 1, $f$ provides a better approximation than any fact in $F_1$ as well (since $F_1\subseteq F_2$). Hence, $D(F_1\cup\{f\},r)=D(F_2\cup\{s\},r)$ and $\Delta U(F_1,f)\geq \Delta U(F_2,f)$ (i.e., utility is sub-modular) since $D(F_1,r)\geq D(F_2,r)$. In case 2, we have $\Delta U(F_2,f)=0$ and therefore $\Delta U(F_1,f)\geq \Delta U(F_2,f)$ (as the utility delta is non-negative). Utility is aggregated as the sum over multiple rows. The sum of sub-modular functions (with positive weights) is sub-modular~\cite{Krause2012}.
\end{IEEEproof}

Next, we upper-bound utility of speech expansions.

\begin{lemma}
For a speech ordering $m$ facts by decreasing single-fact utility, $u_i$ is single-fact utility of the $i$-th fact and $U_i$ aggregate utility of the first $i$ facts, aggregate speech utility is upper-bounded by $U_{i}+(m-i)\cdot u_i$ for any $i$.\label{lm:utilityBound}
\end{lemma}
\begin{proof}
Utility has diminishing returns according to Theorem~\ref{th:diminishing}. Hence, single-fact utility is an upper bound for the increase in speech utility after adding the $i$-th fact (i.e., $U_{i}-U_{i-1}\leq u_i$). Also, as facts are ordered, we have $u_{k}\leq u_i$ for $k\in i,\ldots,m$ and with $U_m=U_{i}+\sum_{k=i+1..m}(U_{k}-U_{k-1})\leq U_{i}+\sum_{k=i+1..m}u_k$, we obtain $U_m\leq U_{i}+u_i\cdot (m-i)$. 
\end{proof}

%Next, we show that Algorithm~\ref{exactAlg} maintains an upper utility bound.

\begin{lemma}
Algorithm~\ref{exactAlg} maintains an upper bound on the utility of speeches $S$ in $S.\mathcal{U}$.\label{lm:upperBound}
\end{lemma}
\begin{proof}
Algorithm~\ref{exactAlg} initializes $\mathcal{U}$ via the single-fact utility (Line~6). In this case, $\mathcal{U}$ represents exact utility of speeches in $S$. Algorithm~\ref{exactAlg} updates $\mathcal{U}$ by adding single-fact utility of new facts after each expansion. Due to diminishing returns (see Theorem~\ref{th:diminishing}), single-fact utility upper-bounds utility increase when adding a fact. %Hence, the upper bound is preserved.
\end{proof}

%via macro $\mathcal{\widetilde{U}}$, adding single-fact utility of the current facts

\begin{theorem}
Pruning preserves an optimal speech.\label{th:pruningPreserves}
\end{theorem}
\begin{proof}
A speech is pruned only if one of the two atoms of macro $\mathcal{P}(b,r)$ evaluates to false. Assume the first atom ($S.\mathcal{U}\geq F.\mathcal{U}$) evaluates to false. This means that a fact with higher single-fact utility follows one with lower single-fact utility. By re-ordering facts, we obtain a speech that satisfies the first atom without changing utility. Assume now that the second atom ($(b-S.\mathcal{U})/r\leq F.\mathcal{U}$) evaluates to false. According to Lemma~\ref{lm:utilityBound}, $S.\mathcal{U}+r\cdot F.\mathcal{U}$ is an upper bound on utility of the completed speech since $r=m-i-1$ (see Algorithm~\ref{exactAlg}, Line~10), $F.\mathcal{U}$ is the single-fact utility of the $i$-th fact ($u_i$ in Lemma~\ref{lm:utilityBound}), and $S.\mathcal{U}$ is an upper bound on the aggregate utility of the first $i-1$ facts (using Lemma~\ref{lm:upperBound}). As $b$ is a lower bound on the optimal utility, having $S.\mathcal{U}+r\cdot F.\mathcal{U}<b$ indicates a fragment that cannot be extended into an optimal speech.
\end{proof}

This immediately implies our main result.

\begin{corollary}
Algorithm~\ref{exactAlg} generates an optimal speech.
\end{corollary}
\begin{proof}
The algorithm considers all possible speeches that are not pruned out. Due to Theorem~\ref{th:pruningPreserves}, pruning preserves optimal speeches. The algorithm calculates exact utility of each remaining speech and therefore identifies an optimum.
\end{proof}

\section{Greedy Algorithm}
\label{greedySec}

The greedy algorithm finds guaranteed near-optimal speech summaries efficiently. We describe the algorithm in Section~\ref{sub:greedyAlg} and prove its properties in Section~\ref{sub:nearOptimality}.

%We present a greedy algorithm for data summarization. We describe the base algorithm in Section~\ref{sub:greedyAlg}. We show in Section~\ref{sub:nearOptimality} that this algorithm already guarantees near-optimal speech summaries. %Finally, we describe extensions in Section~\ref{sub:extensions}.

\subsection{Algorithm}
\label{sub:greedyAlg}

\begin{algorithm}[t]
\renewcommand{\algorithmiccomment}[1]{// #1}
\begin{algorithmic}[1]
\State \Comment{Find approximately optimal summary for}
\State \Comment{relation $R$ using $m$ facts from $F$.}
\Function{ApproximateSummary}{$F,R,m$}
%\State \Comment{Initialize user expectations}
%\State $R\gets\Pi_{\mathcal{D},R}(R)$
\State \Comment{Iterate over speech facts}
\For{$i\gets1,\ldots,m$}
\State \Comment{Calculate aggregate utility for each fact}
\State $U\gets\Gamma_{\sum\mathcal{U},F}(R\Join_{\mathcal{M}}F)$
\State \Comment{Select maximum utility fact}
\State $f_i^*\gets\arg\max_{f\in F}(U)$
\State \Comment{Recalculate user expectation}
\State $R\gets\Pi_{\mathcal{E},R}(R\Join_{\mathcal{M}} f_i^*)$
\EndFor
\State \Comment{Combine max-utility facts}
\State \Return{$f_1^*\Join\ldots\Join f_m^*$}
\EndFunction
\end{algorithmic}
\caption{Greedily add most useful facts to obtain guaranteed near-optimal speech summaries.\label{approxAlg}}
\end{algorithm}

The greedy algorithm generates speeches by iteratively adding facts, starting from an empty speech. In each iteration, it greedily adds the fact that increases utility by the highest amount. This algorithm is more efficient than exhaustive search as it avoids considering fact combinations for speech expansions. Also, this seemingly simple strategy guarantees speeches within a factor of $(1-1/e)$ of the optimal utility.

Algorithm~\ref{approxAlg} shows the associated pseudo-code. It uses the same basic operators as Algorithm~\ref{exactAlg}. Given a relation $R$ to summarize with up to $m$ facts from $F$, the algorithm returns a near-optimal combination of facts. In each iteration (Line~7), the algorithm calculates utility gain of each fact (i.e., added utility when expanding the current speech by that fact). This is realized by a join, pairing facts with data rows within their scope (join condition $\mathcal{M}$), followed by aggregating utility gain for each fact over all rows. After identifying the fact with maximal added utility (Line~9), the algorithm recalculates user expectations based on the expanded speech (Line~11). Here, $\mathcal{E}$ represents an SQL expression that calculates the value, expected by users after listening to the current speech, according to our model (see Section~\ref{modelSec}). The resulting values are stored as a column of the updated relation $R$ (and initialized with the prior). They are used for calculating utility gain of facts in Line~7. Finally, the combination of locally optimal facts is returned.

\begin{example}
We consider our running example, illustrated in Figure~\ref{fig:delays}. We calculate utility under the same assumptions as in Example~\ref{ex:utility}. We consider all facts on average delay describing flights within a specific region or season or both. Considering those facts, the greedy algorithm selects either the fact referencing flights in Winter or the one referencing flights in the North (both tied with a maximal utility of 40). In the second iteration, it will select the other one of the two aforementioned facts (now with a maximal utility gain of 25). Other facts, e.g. referencing flights in the South in Summer, with utility 20, are dominated.
\end{example}

\subsection{Proof of Near-Optimality}
\label{sub:nearOptimality}

%We show that this algorithm guarantees near-optimal utility. 

%Intuitively, this means that adding more facts yields less and less information. 
We prove that Algorithm~\ref{approxAlg} produces near-optimal speeches.

\begin{theorem}
Algorithm~\ref{approxAlg} produces speeches with utility within factor $(1-1/e)$ of the optimum.
\end{theorem}
\begin{proof}
Utility is non-negative, monotone, and sub-modular (according to Theorem~\ref{th:diminishing}). Algorithm~\ref{approxAlg} greedily selects a bounded number of facts for each scope. Doing so guarantees the postulated optimality factor~\cite{Nemhauser1978}.
\end{proof}

\begin{comment}
\subsection{Extensions}
\label{sub:extensions}

To increase readability, we have simplified the pseudo-code, compared to the actual implementation. Next, we give a short overview of how the two differ. First, the pseudo-code summarizes one single data set. In the implementation, we generally summarize batches of data subsets. For that, we consider data subsets that are defined by equality conditions. Given a fixed set of restricted columns, each data row falls into one of those data subsets. Hence, we can efficiently generate summaries for all subsets by calculating utility per fact and per subset via a group-by operation. The total number of queries executed is significantly reduced, compared to treating data subsets separately. %In the voice interface deployment, speech input queries are mapped to specific data subsets whose speeches are generated in the aforementioned way.

Second, our actual implementation supports a more diverse set of fact types. In addition to facts describing averages, we support facts describing minimal and maximal bounds on the distribution of target values in row subsets. We can show sub-modularity for an extended utility model, based on the latter types of facts, as well. The formal guarantees of the presented algorithm therefore generalize to different types of facts. We omit details on extensions due to space restrictions and to increase readability. Both extensions, i.e.\ batch processing and more diverse fact types, can be applied to the algorithms described in the following sections as well.
\end{comment}

\section{Pruning Facts}
\label{pruningSec}

We show how to prune facts early for greedy speech construction. Section~\ref{sub:pruningOverview} gives an overview, Section~\ref{sub:pruning} describes the pruning mechanics in detail, Sections~\ref{sub:costModel} and \ref{sub:pruningOptimizer} provide details on the cost-based optimizer.

%Section~\ref{sub:pruning} describes the general technique according to which facts are pruned. Section~\ref{sub:costModel} introduces a cost model that allows to compare alternative pruning strategies. Finally, Section~\ref{sub:pruningOptimizer} introduces a pruning plan optimizer based on that cost model.

\subsection{Overview}
\label{sub:pruningOverview}

The greedy algorithm must select the fact with maximal utility gain in each iteration (otherwise, the formal guarantees on finding near-optimal speeches do not apply). A naive method calculates utility gain for each fact to determine the maximum. Doing so requires pairing up data rows with facts (i.e., a join in Algorithm~\ref{approxAlg}), an expensive operation. In some cases, we can conclude more efficiently that fact groups do not yield maximal utility. To do so, we compare an upper bound on utility gain of facts (calculated without a join) against utility of other facts. This creates overheads for calculating bounds and for comparisons. Those overheads only pay off if they remove facts from further considerations. Hence, we use cost-based planning to decide if and how to try excluding facts for minimal processing costs. The following example illustrates the high-level principle.

\begin{example}\label{ex:factPruning}
We consider our running example from Figure~\ref{fig:delays}, calculating utility as in Example~\ref{ex:utility}. Assume the greedy algorithm selected the fact stating average delays in Winter (15 minutes) in the first iteration. In the second iteration, our goal is to efficiently identify a fact with maximal utility gain. By summing up absolute differences between expected and actual delays for a specific season (expectation is influenced by priors and the first fact), we obtain upper bounds on utility gain of any fact referencing that season. Similarly, we obtain upper bounds for facts referencing specific regions. For instance, facts referencing Fall have an upper bound of 10 and facts referencing the East cannot increase utility by more than five (deviation between actual and expected delay in the East in Winter). Assume we calculate utility gain of the fact stating average delays in the North (15 minutes) first. Based on its utility gain (25) and the upper bounds, we can exclude all other facts from further consideration. Calculating utility of other facts first is less effective for pruning.
\end{example}

\subsection{Pruning Method}
\label{sub:pruning}

Algorithm~\ref{utilityAlg} describes the pruning mechanism (it replaces Line~7 in Algorithm~\ref{approxAlg}). We prune facts at the granularity of fact groups, characterized by the set of restricted dimension columns (denoted as $R.Dims$ in Algorithm~\ref{utilityAlg}). For instance, in the context of Example~\ref{ex:factPruning}, we consider all facts referencing specific regions (but no specific seasons) as one group. Algorithm~\ref{utilityAlg} first determines the set of fact groups (Line~5) and an optimal pruning strategy (Line~7), as discussed in the following subsections. A pruning strategy consists of a source $S$ and a target $T$ ($F(S)$ and $F(T)$ denote facts associated with the corresponding groups). The pruning source is a set of fact groups whose utility is calculated first. Then, the maximal utility gain of any source fact (denoted as $m$ in Algorithm~\ref{utilityAlg}) is used to prune target facts, based on their upper utility bounds. Upper bounds are calculated by summing up absolute deviation between expectation and correct values (denoted as $\mathcal{D}$) over all data rows (Line~15), grouping by values in dimension columns for a fixed fact group. Clearly, adding a fact can at most decrease error to zero in the data region the fact refers to. This implies the upper bound on utility gain. If the target group is dominated (check in Line~17), we prune not only the target group but also its specializations. The specialization of a target group restricts a strict superset of dimension columns. Specializing a fact reduces its scope to a data subset. Hence, upper utility bounds (obtained by summing deviation over all rows within the scope of a fact) apply to fact specializations as well. Finally, utility is calculated for all remaining fact groups.

\begin{algorithm}[t]
\renewcommand{\algorithmiccomment}[1]{// #1}
\begin{algorithmic}[1]
\State \Comment{Prune efficiently to interesting fact groups from $F$}
\State \Comment{before calculating utility for summarizing relation $R$.}
\Function{Utility}{$F,R$}
\State \Comment{Collect all available fact groups}
\State $G\gets$\Call{PowerSet}{$R.Dims$}
\State \Comment{Determine optimal pruning strategy}
\State $\langle S,T\rangle\gets$\Call{OptPrune}{$F,R$}
\State \Comment{Calculate utility for pruning source}
\State $m\gets\max(\Pi_{\mathcal{U}}(\Gamma_{\sum\mathcal{U},F}(R\Join_{\mathcal{M}}F(S))))$
\State \Comment{Iterate over pruning targets}
\For{$t\in T$}
\State \Comment{Is target unpruned?}
\If{$t\in G$}
\State \Comment{Calculate upper utility bound}
\State $u\gets \max(\Pi_{\mathcal{D}}(\Gamma_{\sum\mathcal{D},t}(R)))$
\State \Comment{Source dominates target?}
\If{$m>u$}
\State \Comment{Prune target and all specializations}
\State $G\gets G\setminus \{g\in G|t\subseteq g\}$
\EndIf
\EndIf
\EndFor
\State \Comment{Return aggregate utility for interesting facts}
\State \Return{$\Gamma_{\sum\mathcal{U},F}(R\Join_{\mathcal{M}}F(G))$}
\EndFunction
\end{algorithmic}
\caption{Reduce overheads for utility calculations by applying optimized fact pruning strategy.\label{utilityAlg}}
\end{algorithm}
\subsection{Cost Model}
\label{sub:costModel}

We introduce a cost model for pruning choices. This model estimates processing cost given pruning source and target. For facts in a group $g$, we denote by $C_U(g)$ the estimated cost of calculating utility of each fact (this requires a join between facts and data rows). By $C_D(g)$, we denote estimated cost for calculating deviation between expected and actual values for row groups (this requires a group-by query without joins). Both estimates can be obtained via the query optimizer cost model. Denoting by $P_g$ the event that group $g$ is pruned, we can estimate data processing cost of Algorithm~\ref{utilityAlg} as

\begin{equation*}
\sum_{s\in S}(C_U(s))+\sum_{t\in T}(C_D(t))+\sum_{g\in G\setminus S}(\Pr(\neg P_g)\cdot C_U(g))
\end{equation*}

. The first term represents cost for calculating bounds based on pruning sources. The second term represents cost of calculating bounds for the pruning targets (we simplify by assuming that all target groups are treated). The last term represents cost of calculating utility for facts that remain after pruning. It depends on the probability $\Pr(\neg P_g)$ that a group $g$ is not pruned.

%\begin{equation*}\Pr(\neg P_g)=\prod_{s\in S}\prod_{t\in T:t\subseteq g}(1-\Pr(P_{s\rightarrow t}))\end{equation*}

A fact group may be pruned if it was a pruning target or if it specializes a pruning target. The following formula covers both possibilities: $\Pr(\neg P_g)=\Pr(\nexists t\in T:t\subseteq g\wedge P_t)$. It can be expanded into $\Pr(\nexists s\in S,t\in T:t\subseteq g\wedge P_{s\rightarrow t})$, denoting by $P_{s\rightarrow t}$ the event that the upper utility bound for facts $t$ is below the lower bound for facts in $s$. We simplify by assuming independence between different pruning outcomes, obtaining $\Pr(\neg P_g)=\prod_{s\in S}\prod_{t\in T:t\subseteq g}(1-\Pr(P_{s\rightarrow t}))$. Finally, we estimate probability for $P_{s\rightarrow t}$. Modeling utility bounds per fact as a sum over i.i.d. random variables representing utility per row, it approaches a normal distribution as the number of rows grows (due to the Central Limit Theorem). We assume that per-row utility follows the same distribution, independently of the fact group and of the type of bound calculated. Hence, the per-fact utility distribution only depends on the number of rows within its scope. We simplify by assuming a uniform distribution of rows over dimension values. Then, the number of rows that are within the scope of a fact is inversely proportional to the number of facts in the fact group. We can estimate the number of facts in group $s$ and $t$, denoted by $M(s)$ and $M(t)$ in the following, by referring to query optimizer statistics. The number of facts simply equals the number of distinct value combinations in the dimension columns they restrict. Finally, we assume that the variance of the per-fact utility distribution is fixed and given by $\sigma^2$. Under those assumptions, we can express pruning probability by comparing two normal distributions:

\begin{small}
\begin{equation*}
\Pr(P_{s\rightarrow t})=\Pr(u_s>u_t|u_s\sim\mathcal{N}(\frac{1}{M(s)},\sigma^2),u_t\sim\mathcal{N}(\frac{1}{M(t)},\sigma^2))
\end{equation*}
\end{small}

. Here, $\mathcal{N}(\mu,\sigma^2)$ designates the normal distribution with mean $\mu$ and variance $\sigma^2$. This cost model is based on various simplifying assumptions. Nevertheless, we will show experimentally that it is sufficient to avoid bad pruning plans.

\subsection{Pruning Optimization}
\label{sub:pruningOptimizer}

\begin{algorithm}[t]
\renewcommand{\algorithmiccomment}[1]{// #1}
\begin{algorithmic}[1]
\State \Comment{Generate plans for pruning facts $F$ on relation $R$.}
\Function{Plans}{$F,R$}
\State \Comment{Initialize set of plan candidates}
\State $P\gets\emptyset$
\State \Comment{Collect available fact groups}
\State $G\gets$\Call{PowerSet}{$R.Dims$}
\State \Comment{Iterate over pruning sources}
\For{$S\subseteq G:\nexists s\in S,g\in G\setminus S:M(g)<M(s)$}
\State \Comment{Initialize pruning target set}
\State $T\gets\emptyset$
\State \Comment{Initialize pruning targets left}
\State $L\gets G\setminus S$
\State \Comment{Iterate until no targets left}
\While{$L\neq\emptyset$}
\State \Comment{Select next pruning target}
\State $t\gets\arg\max_{t\in L}(H(t,S,L))$
\State \Comment{Add to pruning targets}
\State $T\gets T\cup\{t\}$
\State \Comment{Add corresponding plan candidate}
\State $P\gets P\cup\{\langle S,T\rangle\}$
\State \Comment{Discard specialization groups}
\State $L\gets L\setminus\{g\in G|t\subseteq g\}$
\EndWhile
\EndFor
\State \Return{$P$}
\EndFunction
\end{algorithmic}
\caption{The optimal pruning plan is selected from plan candidates generated by this algorithm.\label{planAlg}}
\end{algorithm}

We select pruning plans based on the cost model presented in the last subsection. Function~\textproc{OptPrune}, used in Algorithm~\ref{utilityAlg}, returns the minimum cost plan among a set of candidates. The number of candidate plans grows exponentially in the number of fact groups. To reduce optimization overheads, we use several heuristic to obtain a smaller set of candidate plans, calculated by Algorithm~\ref{planAlg}. 

When selecting pruning sources, Algorithm~\ref{planAlg} prioritizes fact groups with few member facts. For such groups, the expected utility is higher as each fact tends to cover more rows. For each possible pruning source, Algorithm~\ref{planAlg} considers multiple possible pruning target sets. Pruning targets are selected according to a heuristic function $H$. Given pruning sources $S$ and remaining groups $L$, the value of a pruning target $t$ is given by $H(t,S,L)=\Pr(P_t)\cdot|\{l\in L:t\subseteq l\}|$. This function estimates the expected number of fact groups that can be removed after calculating bounds for target $t$. After selecting the next pruning target, all fact groups specializing the target group are removed from further consideration. If the target group can be successfully pruned, those fact groups would be implicitly pruned as well. Hence, we exclude them from further consideration. Each combination of a source and target set yields a new plan candidate. %At the end of the function, all generated plans are returned. 

\section{Complexity Analysis}
\label{analysisSec}

We analyze the complexity class of speech summarization. 

\begin{theorem}
Speech summarization is NP-hard.
\end{theorem}
\begin{IEEEproof}
We use a reduction from set cover. An instance of set cover is defined by a universe $U$, a set $S$ of subsets of $U$, and an integer $m$. The decision variant asks whether $U$ can be covered with $m$ elements from $S$. We reduce to speech summarization as follows. We use a relation $R$ with one row for each element in $U$. For each subset $s\in S$ (with $s\subseteq U$), we introduce a candidate fact $F_s$ with value 1. That fact restricts the dimension columns such that exactly rows $R_s\subseteq R$, associated with elements from $s$, fall within its scope. We introduce one relation column $C_s$, associated with $s$, such that only rows $R_s$ are set to a unique value $v_s$ in $C_s$ and $F_s=\langle \{\langle C_s,v_s\rangle\},1\rangle$. We set a uniform prior $P(r)=0$ for all rows, the target value is uniformly set to one. We can achieve a deviation of zero if and only if each row falls within the scope of at least one fact. If the optimal speech with $m$ facts has deviation zero then $U$ can be covered with $m$ sets from $S$. The reduction has polynomial time complexity.
\end{IEEEproof}

%Hence, unless $P=NP$, no polynomial time algorithm can generate optimal speech summaries. 

Next, we analyze time and space complexity of the proposed algorithms. We denote by $n=|R|$ the number of rows in the relation to summarize, by $m$ the maximal number of facts to select, and by $k=|F|$ the number of fact candidates. We measure time and space complexity by the number of rows (or row combinations) processed or stored. %as the number of rows processed (rather than the amount of data read or written).

\begin{theorem}
Algorithm~\ref{exactAlg} has time complexity $O(n\cdot \binom{k}{m})$.\label{th:exactTime}
\end{theorem}
\begin{IEEEproof}
Under worst case assumptions, only the first of the two pruning condition (eliminating redundant fact permutations) is effective. Then, the number of partial speeches grows quickly in the number of iterations. This means that the accumulated cost of prior join operations is negligible, compared to the cost of the final join. Selecting $m$ facts out of $k$ candidates yields $O(\binom{k}{m})$ possibilities. The last join pairs potentially optimal speeches with $n$ data rows.  Assuming nested loops joins, its complexity is in $O(n\cdot\binom{k}{m})$.
\end{IEEEproof}

\begin{theorem}
Algorithm~\ref{approxAlg} has time complexity $O(m\cdot n\cdot k)$.
\end{theorem}
\begin{IEEEproof}
The operation with dominant complexity is the join between data rows and facts. A nested loops join has complexity $O(n\cdot k)$. We perform $O(m)$ iterations.
\end{IEEEproof}

\begin{comment}
As outlined in Section~\ref{sub:greedyAlg}, we can alternatively partition facts and create indices to speed up the join between data and candidate facts. This can be beneficial if the number of iterations is high. 

\begin{theorem}
Algorithm~\ref{approxAlg} with index nested loops joins has time complexity $O(n\cdot 2^d+m\cdot n\cdot k\cdot\sigma)$ time.
\end{theorem}
\begin{IEEEproof}
The operation with dominant time complexity is the join between data rows and candidate facts. Partitioning facts before iterating is in $O(k)$ time. Creating an index on dimension columns for each fact partition is in $O(n\cdot 2^d)$ time. Join complexity reduces to $n\cdot k\cdot \sigma$ as matching tuples can be retrieved efficiently, we perform $O(m)$ such joins. 
\end{IEEEproof}

Next, we analyze the complexity of calculating guaranteed optimal speeches.
\end{comment}

%Finally, we analyze space complexity.

\begin{theorem}
Algorithm~\ref{exactAlg} is in $O(\binom{k}{m})$ space.
\end{theorem}
\begin{IEEEproof}
The relation with dominant space requirements is the one storing candidate speeches. It has $O(\binom{k}{m})$ rows (using a similar reasoning as for Theorem~\ref{th:exactTime}).
\end{IEEEproof}

\begin{theorem}
Algorithm~\ref{approxAlg} is in $O(n+m+k)$ space.
\end{theorem}
\begin{IEEEproof}
Besides the input, the algorithm stores a user expectation associated with each row, a utility value for each fact, and one optimal fact for each iteration. Respectively, space consumption is in $O(n)$, $O(k)$, and $O(m)$.
\end{IEEEproof}

% (updated after each iteration)
%Finally, we analyze scaling in the table dimensions. 

Finally, we bound the number of facts and queries considered by the system. Let $d$ be the number of dimension and $t$ the number of target columns, and $l$ the number of predicates used in each query and in each fact.

\begin{theorem}
The number of facts is in $O(\binom{d}{l}\cdot n^l)$.
\end{theorem}
\begin{IEEEproof}
A fact is defined by picking $l$ of $d$ columns and one of $O(n)$ possible equality predicates for each column.
\end{IEEEproof}

\begin{theorem}
The number of queries is in $O(t\cdot\binom{d}{l}\cdot n^l)$.
\end{theorem}
\begin{IEEEproof}
A query is defined by picking one of $t$ target columns, $l$ of $d$ columns for placing predicates, and one of $O(n)$ possible predicates for each column.
\end{IEEEproof}

%Furthermore, we denote by $\sigma$ the probability that a row falls within the scope of a fact and by $d$ the number of dimension columns. 

%Altogether, the analysis indicates that the greedy algorithm scales significantly better in the number of speech facts. 

\section{Experimental Evaluation}
\label{experimentsSec}

We evaluate our algorithms in various experimental settings.

%We describe the setup in Section~\ref{setupSub} and compare speech generation methods in Section~\ref{sub:computational}. In Section~\ref{sub:userStudies}, we justify the speech quality model used for optimization via user studies, and perform an end-to-end comparison to visual interfaces. In Section~\ref{sub:deployment}, we analyze results from a public deployment of our voice interface. \textcolor{red}{Finally, we compare against baselines in Section~\ref{sub:vsBaselines}.}

\subsection{Experimental Setup}
\label{setupSub}

\begin{table}[t]
\caption{Overview of data sets used for experiments.\label{dataSetsTable}}
\begin{tabular}{p{2cm}p{1.5cm}p{1.5cm}p{1.75cm}}
\toprule[1pt]
\textbf{Data Set} & \textbf{Size} & \textbf{\#Dims} & \textbf{\#Targets}\\
\midrule[1pt]
ACS NY & 2~MB & 3 & 6\\
\midrule
Stack Overflow & 197~MB & 7 & 6 \\
\midrule
Flights & 565~MB & 6 & 1 \\
\midrule
Primaries & 6~MB & 5 & 1\\
\bottomrule[1pt]
\end{tabular}
\end{table}

Table~\ref{dataSetsTable} summarizes the data sets used for the following experiments. We consider an extract from the American Community Survey (ACS) focused on disability statistics, results of the 2019 Stack Overflow Developer survey\footnote{\url{https://insights.stackoverflow.com/survey/2019}}, a data set on flight delays\footnote{\url{https://www.kaggle.com/usdot/flight-delays}} (frequently used to evaluate OLAP interfaces~\cite{Trummer2018a}), and a data set on the democratic primaries\footnote{\url{https://data.fivethirtyeight.com/}}. Table~\ref{dataSetsTable} reports the number of dimensions and the total number of target columns. Unless noted otherwise, we generate speeches with three facts (prior work shows that user retention decreases sharply after three facts~\cite{Trummer2018}), considering all facts restricting up to two dimension columns. For all experiments, we use the average value in the target column as a (constant) prior. We used Postgres~9.5 as relational database, running on a t3.2xlarge EC2 instance with eight virtual cores and 350~GB of EBS volume. The operating system is Ubuntu Linux~18.04.

%All of the following performance experiments were executed on a MacBook Air laptop with 8~GB main memory and 1.8 GHz Intel Core i5 CPU. Postgres~11.5 was used as relational processing engine. We configured the maximal speech length to three facts, a value that seems reasonable given limitations of the human short term memory~\cite{Cowan2010}. 

\subsection{Comparing Pre-Processing Methods}
\label{sub:computational}

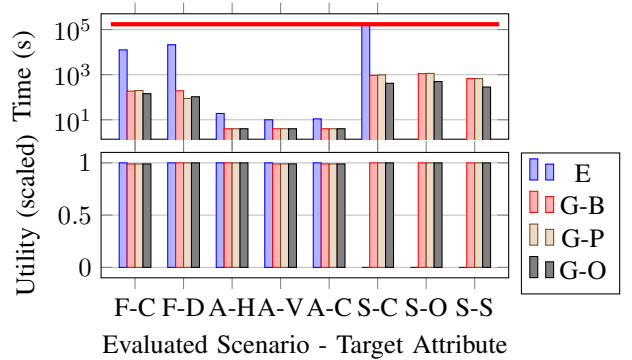
\begin{figure}[t]
\centering
\begin{tikzpicture}
\begin{groupplot}[group style={group size=1 by 2, x descriptions at=edge bottom, vertical sep=5pt}, xlabel={Evaluated Scenario - Target Attribute}, ymode=log, ybar=0pt, width=7cm, height=3.25cm, ymajorgrids, xtick=data, xticklabels={F-C, F-D, A-H, A-V, A-C, S-C, S-O, S-S}]
\nextgroupplot[ylabel={Time (s)}, bar width=3pt]
\addplot coordinates {(1,12622) (2,21330) (3,19) (4,10) (5,11) (6,172800) (7,0) (8,0)};
\addplot coordinates {(1,187) (2,193) (3,4) (4,4) (5,4) (6,933) (7,1111) (8,671)} ;
\addplot coordinates {(1,199) (2,89) (3,4) (4,4) (5,4) (6,979) (7,1141) (8,668)};
\addplot coordinates {(1,144) (2,105) (3,4) (4,4) (5,4) (6,416) (7,495) (8,284)} ;
\draw[red, ultra thick] (axis cs:0.5,172800) -- (axis cs:8.5,172800);
\nextgroupplot[ylabel={Utility (scaled)}, bar width=3pt, ymode=linear, legend columns=1, legend entries={E,G-B,G-P,G-O}, legend pos=outer north east]
\addplot coordinates {(1,1) (2,1) (3,1) (4,1) (5,1) (6,0) (7,0) (8,0)};
\addplot coordinates {(1,0.99) (2,1) (3,1) (4,0.99) (5,0.99) (6,1) (7,1) (8,1)};
\addplot coordinates {(1,0.99) (2,1) (3,1) (4,0.99) (5,0.99) (6,1) (7,1) (8,1)};
\addplot coordinates {(1,0.99) (2,1) (3,1) (4,0.99) (5,0.99) (6,1) (7,1) (8,1)};
%\addplot coordinates {(1,82562469) (2,)} ;
\end{groupplot}
\end{tikzpicture}

%\ref{performanceLegend} legend to name=performanceLegend
\caption{Performance comparison of presented algorithms in three scenarios (the red line marks the timeout).\label{fig:performance}}
\end{figure}

We compare four methods for generating speeches. We compare the exact algorithm (see Section~\ref{exactSec}), E in the following plots, against the greedy algorithm G-B in its base version (see Section~\ref{greedySec}), and with two other greedy variants, G-P and G-O, that use different types of fact pruning. Algorithm~G-P uses the fact pruning method described in Section~\ref{sub:pruning} with a simple pruning strategy. It uses all fact groups for pruning in the same order in which they are considered by Algorithm~\ref{planAlg}. Algorithm~G-O uses the pruning optimizer and cost model, introduced in Sections~\ref{sub:costModel} and \ref{sub:pruningOptimizer}. Figure~\ref{fig:performance} shows computation time and average utility of generated speeches for three scenarios and several target columns. We consider cancellation and delay for flights (F-C, and F-D), prevalence of hearing loss, visual impairment, and cognitive impairment for the ACS data (A-H, A-V, and A-C), and competence, optimism, and job satisfaction for the Stackoverflow data set (S-C, S-O, and S-S). We set a per-scenario timeout of 48 hours of processing time. We measure speech utility according to the model from Section~\ref{modelSec} and scale to one for each summarization problem instance. For each test case, we generated speeches for all queries with up to two equality predicates on the dimensions. With this configuration, each approach generated about 11,300 speeches for the Stackoverflow scenario, about 8,500 speeches about flights, and about 2,900 speeches for the ACS data.

%We measure speech utility by the reduction of the deviation between expectations and actual values, before and after listening to a speech (estimated via the model introduced in Section~\ref{modelSec}). For a fixed test case, we scale speech utility to the maximal speech utility measured by any algorithm that finished within the timeout. 
Exact optimization is possible within the time limit for flights and the ACS data (taking up to six hours for flight delays). For Stackoverflow, optimization for the first attribute takes more than 48 hours. This correlates with the number of facts reaching 3,700 facts per data subset in case of Stackoverflow (it is 1,300 facts for flights data and 764 facts for ACS). Greedy optimization is orders of magnitude faster. It produces speeches that are of comparable quality to exact optimization (the minimal average utility per test case is 98\%, far above the theoretical minimum of around 66\%). Fact pruning is only helpful if combined with optimization. Naive pruning may even increase computational overheads (due to added cost for establishing and comparing utility bounds). The optimization method presented in Section~\ref{sub:pruningOptimizer} reduces total computation time from 3,107 seconds (G-B) to 1,456 seconds (G-O). Naive pruning improves only slightly over the base version (3,088 seconds).

%Combined 26.1 MB

\begin{figure}[t]
\centering
\begin{tikzpicture}
\begin{groupplot}[group style={group size=3 by 2, ylabels at=edge left}, width=2.75cm, height=3cm, ylabel={Time (s)}, ymajorgrids]
\nextgroupplot[title={A-H - Length}, ylabel={Time (ms)}]
\addplot coordinates {(1,2642) (3,3739) (5,5157)};
\addplot coordinates {(1,2442) (3,3657) (5,4801)};
\nextgroupplot[title={F-C - Length}]
\addplot coordinates {(1,61) (3,144) (5,232)};
\addplot coordinates {(1,84) (3,187) (5,289)};
\nextgroupplot[title={S-O - Length}]
\addplot coordinates {(1,296) (3,495) (5,697)};
\addplot coordinates {(1,503) (3,1111) (5,1715)};
\nextgroupplot[title={S-O - Dims}, legend entries={G-O, G-P}, legend columns=2, legend to name=scalabilityLegend]
\addplot coordinates {(1,129) (2,495) (3,1816)};
\addplot coordinates {(1,128) (2,1111) (3,4627)};
\nextgroupplot[title={A-H - Dims}]
\addplot coordinates {(1,2862) (2,3739) (3,4062)};
\addplot coordinates {(1,2650) (2,3657) (3,3878)};
\nextgroupplot[title={F-C - Dims}]
\addplot coordinates {(1,42) (2,144) (3,337)};
\addplot coordinates {(1,42) (2,187) (3,500)};
\end{groupplot}
\end{tikzpicture}

\ref{scalabilityLegend}
\caption{Scaling speech length and fact dimensions.\label{fig:scalability}}
\end{figure}
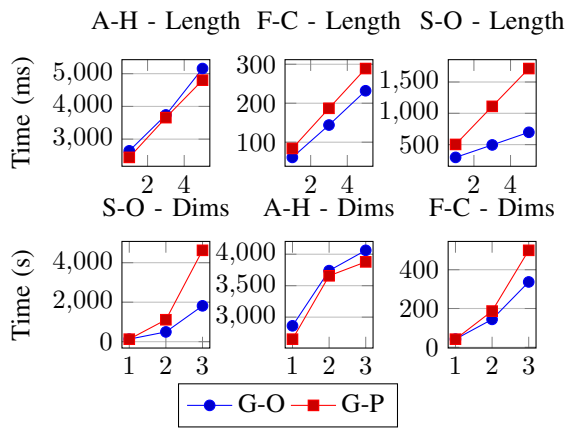

Finally, we assess the scalability of G-O and G-B in terms of the speech length (number of selected facts) and the maximal number of dimensions mentioned per single fact. Figure~\ref{fig:scalability} reports corresponding results. Scaling is more graceful when increasing speech length, compared to the number of dimensions per fact. G-O reduces overheads, compared to G-P. %In line with the results of our complexity analysis, computation time scales linearly in the speech length. Performance decreases faster when increasing the number of columns facts refer to. This seems logical as the number of facts may increase exponentially in the number of referred columns. Also, the relative impact of optimized fact pruning increases with increasing number of facts. Overall, the methods presented in the previous sections are effective at reducing computational overheads.

\subsection{User Studies}
\label{sub:userStudies}

%\begin{comment}
\begin{figure}[t]
\centering
\begin{tikzpicture}
\begin{groupplot}[group style={group size=2 by 1, horizontal sep=1.25cm}, width=3.6cm, height=3.25cm, xticklabels={Precise, Good, Complete, Informative}, ymajorgrids, xticklabel style={rotate=60, font=\small}, xtick={1,2,3,4}, xmin=0.5, xmax=4.5]
% precise, good, complete, informative
\nextgroupplot[ybar=0pt, bar width=4pt, ylabel={Nr.\ Wins}]
\addplot coordinates {(1, 31) (2,31) (3,31) (4,34)};
\addplot coordinates {(1,23) (2,33) (3,25) (4,22)};
\addplot coordinates {(1,44) (2,42) (3,37) (4,50)};
\nextgroupplot[ybar=0pt, bar width=4pt, ylabel={Quality}, legend entries={Worst, Medium, Best}, legend columns=1, legend pos=outer north east]
\addplot coordinates {(1, 6.7) (2,6.2) (3,6.5) (4,6.4)};
\addplot coordinates {(1,6.71) (2,6.7) (3,6.7) (4,6.6)};
\addplot coordinates {(1,6.9) (2,6.7) (3,6.8) (4,6.9)};
\end{groupplot}
\end{tikzpicture}
%\ref{userPrefLegend}legend to name=userPrefLegend
\caption{Preferences of AMT workers correlate with speech quality model.\label{fig:userPref}}
\end{figure}
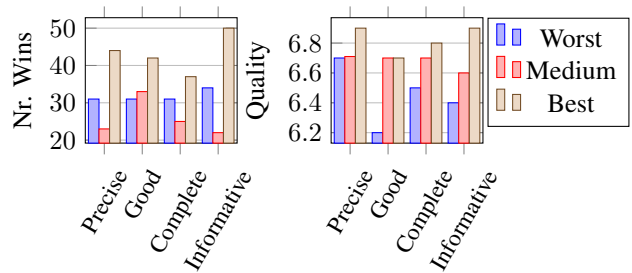
%\end{comment}

We validate our model of speech quality (based on our model of user expectations) through a series of experiments on the Amazon Mechanical Turk (AMT) platform. We generally require a worker acceptance rate of at least 75\% and pay 10 cents per human intelligence task (HIT). Speeches were generated using voice ``Salli'' on \url{TTSMP3.com}. First, for the flights and ACS data set, we generated 100 speeches by randomly selecting facts and ranked them according to our quality model. For both data sets, we selected three speeches (best ranked, worst ranked, and median) for the following study. We asked 50 crowd workers to compare speeches and to rate each on a scale from one to ten according to the criteria ``Precise'', ``Good'', ``Complete'', and ``Informative''. Comparing alternative descriptions of the same data gives crowd workers a base for evaluating the quality of specific speeches. Figure~\ref{fig:userPref} reports average ratings and the number of times, the corresponding speech won a relative comparison. Speech quality, measured according to the model used for optimization, correlates with user preferences. 

%Out of 600 created HITs, we received 596 replies until the deadline.

%Task created on June 19th (Friday) at 4 PM Eastern Time. Acceptance rate of at least 75\%. 10 cents per HIT, requested 500 answers er question (received 596/600 answers until deadline). Average time per assignment: 3 minutes 5 seconds. Relative difference between best and worst cost: 2.6 (salary), 2.23 (flights), 2.22 (acs). Generated facts considering only two high-level dimensions, measured precision compared to breakdown. Generated 100 random speeches and ranked them according to cost model. Average distance between highest and lowest mark was 2.47 (i.e., crowd workers did not fully exploit the scale).

\begin{figure}[t]
\centering
\begin{tikzpicture}
\begin{groupplot}[xlabel={Borough}, ylabel={Count}, group style={group size=3 by 1, ylabels at=edge left, horizontal sep=20pt}, width=3.35cm, xmin=-0.75, xmax=4.75, xticklabels={Brooklyn, Manhattan, Queens, St.\ Island, Bronx}, xtick={0, 1, 2, 3, 4}, xticklabel style={font=\small, rotate=90}, ymajorgrids]
\nextgroupplot[ybar=0pt, bar width=3pt, title={Teenagers}]
\addplot table[x expr=\pgfplotstablerow, y index=0, col sep=tab, header=false] {plots/acsuserstudy/acsRandEstimatesY.csv};
\addplot table[x expr=\pgfplotstablerow, y index=0, col sep=tab, header=false] {plots/acsuserstudy/acsOptEstimatesY.csv};
\addplot table[x expr=\pgfplotstablerow, y index=0, col sep=tab, header=false] {plots/acsuserstudy/acsValuesY.csv};
\nextgroupplot[ybar=0pt, bar width=3pt, title={Adults}]
\addplot table[x expr=\pgfplotstablerow, y index=0, col sep=tab, header=false] {plots/acsuserstudy/acsRandEstimatesA.csv};
\addplot table[x expr=\pgfplotstablerow, y index=0, col sep=tab, header=false] {plots/acsuserstudy/acsOptEstimatesA.csv};
\addplot table[x expr=\pgfplotstablerow, y index=0, col sep=tab, header=false] {plots/acsuserstudy/acsValuesA.csv};
\nextgroupplot[ybar=0pt, bar width=3pt, legend to name=vispLegend, legend entries={Worst, Best, Correct}, legend style={font=\small}, title={Elders}, legend columns=3]
\addplot table[x expr=\pgfplotstablerow, y index=0, col sep=tab, header=false] {plots/acsuserstudy/acsRandEstimatesE.csv};
\addplot table[x expr=\pgfplotstablerow, y index=0, col sep=tab, header=false] {plots/acsuserstudy/acsOptEstimatesE.csv};
\addplot table[x expr=\pgfplotstablerow, y index=0, col sep=tab, header=false] {plots/acsuserstudy/acsValuesE.csv};
\end{groupplot}
\end{tikzpicture}

\ref{vispLegend}
\caption{Comparing worker estimates for visual impairment prevalence, after listening to worst or best speech description, to correct values.\label{fig:visp}}
\end{figure}
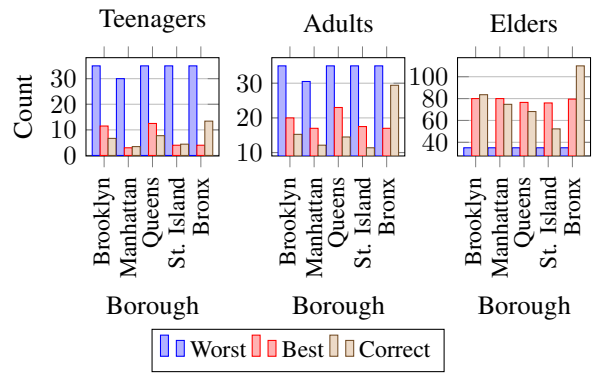

\begin{table}[t]
\caption{Comparing two alternative speech descriptions.\label{twoSpeechesTable}}
\begin{tabular}{p{1.25cm}p{6cm}}
\toprule[1pt]
\textbf{Worst Speech} & About 30 out of 1000 persons in Manhattan identify as visually impaired. It is 35 for Brooklyn. It is 35 overall. \\
\midrule
\textbf{Best Speech} & About 80 out of 1000 elder persons identify as visually impaired. It is 17 for adults. It is 3 for teenagers in Manhattan. \\
\bottomrule[1pt]
\end{tabular}
\end{table}

For the next study, we used the best and worst ranked speech from the ACS scenario, shown in Table~\ref{twoSpeechesTable}. We asked AMT workers to estimate 15 data points, characterized by a borough of New York City and an age group, according to two speeches. We created 20 HITs per data point and per speech, i.e.\ 600 HITs in total (599 of them were answered before the deadline). Figure~\ref{fig:visp} shows how median estimates, based on worst and best ranked speech, correlate with accurate values. Clearly, deviation between estimates and accurate values correlates with speech quality, according to our model.

%111 distinct workers, collected 20 answers by different workers per question, 10 cents per task, minimum acceptance rate of 75\%, obtained answers for 599/600 HITs. Started on June 20th (Saturday) at 7 PM Eastern time. Average duration was 189 seconds. 

\begin{comment}
\begin{figure}[t]
\begin{tikzpicture}
\begin{axis}[width=5.5cm, height=3cm, xlabel={Error}, ylabel={Count}, legend pos=outer north east, legend entries={Worst Speech, Best Speech}, legend columns=1, ymajorgrids]
\addplot+[const plot, no marks, ultra thick] table[x index=0, y index=1, col sep=tab] {plots/acsuserstudy/randomerror.csv};
\addplot+[const plot, no marks, ultra thick] table[x index=0, y index=1, col sep=tab] {plots/acsuserstudy/optimalerror.csv};
\end{axis}
\end{tikzpicture}
\caption{Error distribution of worker estimates after listening to different speeches.\label{fig:acsError}}
\end{figure}

Median error after worst speech is 27.3, 12.4 after best speech.

Average error is 33.8 after worst and 25.5 after best speech.
\end{comment}

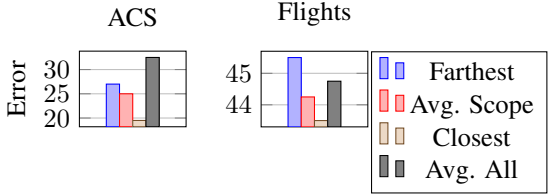
\begin{figure}[t]
\centering
\begin{tikzpicture}
\begin{groupplot}[group style={group size=2 by 1}, xtick=\empty, width=3cm, ybar=0pt, ymajorgrids]
\nextgroupplot[bar width=5pt, title={ACS}, ylabel={Error}]
\addplot coordinates {(1,27)};
\addplot coordinates {(1,25)};
\addplot coordinates {(1,19.5)};
\addplot coordinates {(1,32.5)};
\nextgroupplot[bar width=5pt, legend pos=outer north east, legend entries={Farthest, Avg.\ Scope, Closest, Avg.\ All}, legend columns=1, title={Flights}]
\addplot coordinates {(1,45.5)};
\addplot coordinates {(1,44.25)};
\addplot coordinates {(1,43.5)};
\addplot coordinates {(1,44.75)};
\end{groupplot}
\end{tikzpicture}
\caption{Error when using different models to predict how crowd workers process conflicting facts.\label{fig:intuition}}
\end{figure}

% (for each dimension, the value associated with minimal and maximal average values in the target column was chosen)

%160 HITs in total. All were answered until the deadline

Next, we study how crowd workers resolve conflicting information in a speech description. We consider flight statistics and ACS data on visual impairment. For each data set, crowd workers were given four facts referring to two dimension columns respectively (season and flight time of day for flights, borough and age group for ACS). E.g., in case of the ACS data, crowd workers were given visual impairment statistics for Staten Island and the Bronx, as well as for children and elder persons. We asked crowd workers to estimate prevalence for all four possible combinations of values mentioned in the facts. Workers must therefore prioritize conflicting information between the two respective attributes. We created 20 HITs for each possible combination for both scenarios, all of which were answered. Figure~\ref{fig:intuition} shows the results. For each question, two out of the four facts are relevant (``within scope'', according to the terminology from Section~\ref{modelSec}) and propose a target value. Figure~\ref{fig:intuition} compares four different models of predicting user behavior: estimating the proposed value that is farthest from the accurate one, the one that is closest, using the average of values that appear in relevant facts, and using the average of all values that appear in (relevant or irrelevant) facts. The figure shows the median error between values suggested by crowd workers, and expectations according to the compared model. Using the closest values that appears in relevant facts, yields the best approximation. %This is consistent with the model proposed in Section~\ref{modelSec}. %It shows that users are able, at least in certain scenarios, to identify the most relevant out of multiple facts, based on commonsense knowledge. 

\begin{figure}[t]
\centering
\begin{tikzpicture}
\begin{groupplot}[group style={group size=2 by 1, horizontal sep=1.25cm}, width=3.25cm, height=3.25cm, xmajorgrids, ymajorgrids, xmin=0, ymin=0]
\nextgroupplot[xmax=60, ymax=60, xlabel={Vocal Time (s)}, ylabel={Visual Time (s)}]
\addplot+[only marks] table[col sep=tab, x=VocalAvgTime, y=VisualAvgTime] {plots/ConnorStudies/summary.tsv};
\draw[dashed, thick] (axis cs:0,0) -- (axis cs:60,60);
\nextgroupplot[xmax=10, ymax=10, xlabel={Vocal Eval.}, ylabel={Visual Eval.}]
\addplot+[only marks] table[col sep=tab, x=VocalEval, y=VisualEval]  {plots/ConnorStudies/summary.tsv};
\draw[dashed, thick] (axis cs:0,0) -- (axis cs:10,10);
%\addplot+[only marks] coordinates {()};
\end{groupplot}
\end{tikzpicture}
\caption{User study comparing visual to voice query interfaces.\label{fig:visualVsVoice}}
\end{figure}
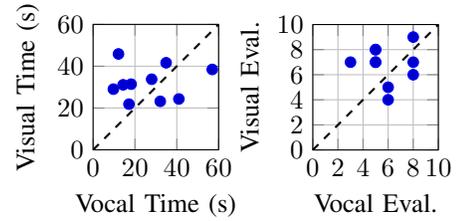

We primarily target use cases where visual interfaces are unavailable, e.g.\ due to limitations of the device (e.g., smart speakers without screen). Nevertheless, we performed a user study over Zoom in which we compare our voice interface against a public interface for visual data analysis\footnote{\url{https://www.visualizefree.com/}}. We recruited 10 undergraduate students (nine of them in computer science) for this study. Participants used their smart phones to access the Stackoverflow data set via the voice interface. We gave participants five minutes to familiarize themselves with each interface and data set. %Before starting the study, we gave participants time (five minutes per interface) to familiarize themselves with both interfaces and the data set.

%, particularities of the environment (e.g., data access while visual attention is focused elsewhere~\cite{Lyons2016}), or other reasons (e.g., visually impaired users). 
% Hence, our goal is not to outperform visual interfaces for data analysis. 

Figure~\ref{fig:visualVsVoice} summarizes the main results. First, we timed participants trying to answer three questions per interface (we offered a bonus based on their average times over all interfaces). Questions were generated by randomly selecting two equality predicates (with uniform distribution) and asking for an average value in the corresponding data subset. We show median values over three questions on the left side of Figure~\ref{fig:visualVsVoice}. The majority of users were slightly faster using the voice interface. Next, we allowed users to freely explore the data to find interesting facts. We asked them to evaluate overall usability of each interface on a scale from one to ten. We show the results on the right-hand side of Figure~\ref{fig:visualVsVoice}. %Overall, six users preferred the visual interface while four users preferred the voice interface (some data points overlap in the figure). %Given our goal of complementing (rather than replacing) visual interfaces, we consider those results promising.

%On the other side, we observed the maximal time per question for the voice interface (147 seconds due to repeated errors in speech recognition). 

\subsection{Public Deployment}
\label{sub:deployment}

\begin{table}[t]
\caption{Classification of last 50 voice requests for three public Cloud deployments.\label{tab:classification}}
\centering
\begin{tabular}{llll}
\toprule[1pt]
\textbf{Request Type} & \textbf{Primaries} & \textbf{Flights} & \textbf{Developers} \\
\midrule[1pt]
Help & 17 & 9 & 4 \\
\midrule
Repeat & 3 & 0 & 0 \\
\midrule
S-Query & 16 & 12 & 13 \\
\midrule
U-Query & 1 & 5 & 16 \\
\midrule
Other & 13 & 24 & 17 \\
\bottomrule[1pt]
\end{tabular}
\end{table}

\begin{figure}[t]
\centering
\subfigure[Queries by complexity (0, 1, 2 predicates from dark blue to yellow).\label{fig:queryComplexity}]{
\begin{tikzpicture}
\pie[radius=0.75,text=legend,sum=auto , after number=]{15/0-Preds, 47/1-Preds, 1/2-Preds}
\end{tikzpicture}
}
\subfigure[Types (retrieval, comparison, extremum from dark blue to yellow).\label{fig:queryType}]{
\begin{tikzpicture}
\pie[radius=0.75,sum=auto , text=legend, after number=]{49/Ret., 6/Cmp., 8/Ext.}
\end{tikzpicture}
}
\caption{Classifying queries by size and type.}
\end{figure}
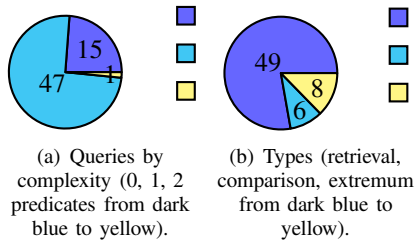

Using our greedy approach, we made three data sets publicly available for voice querying via the Google Assistant platform, namely the Stackoverflow developer survey, flight statistics, and data on the democratic primaries. The first two data sets have been available since January 2020, the last one was available over a period of two months during the democratic primary season. We generated speeches using the approach presented in this paper. At the time of writing, we received about 1,500 unique voice queries. We analyzed the last 50 voice queries for each of the three data sets.

% (e.g., accessible via all Android phones and tablets). Those data sets are the Stackoverflow developer survey, flight statistics, and data on the democratic primary.

% (thereby avoiding overlap with queries issued during our user studies)

% (not considering voice commands for invoking the application itself)

%The resulting speeches take up only 26~MB of memory, thereby fitting into main memory of a standard Google App engine instance. 

First, we classified the 150 queries into help requests, requests for repeating the last output, supported and unsupported (data access) queries (S-Query and U-Query), and other requests. Table~\ref{tab:classification} shows the results. Help requests are relatively common and show the importance of corresponding mechanisms (which our application provides). We take the low number of ``repeat output'' requests as a sign that voice output was mostly understandable. %The number of ``other'' requests is by far the highest for the flights data set. In particular, we observed multiple requests from users who tried to book flights via our interface.

%The ratio of supported queries is significantly higher for flights and primaries than for the Stackoverflow developer survey. We suspect that the latter data set, more than the others, attracts users who are used to complex data access queries. We plan to verify this hypothesis with further studies. 

Considering data access queries, our query model covers about two thirds of incoming queries. The unsupported queries include queries asking for maxima or for relative comparisons between two data subsets (e.g., ``make a comparison between job satisfaction between men and women''), as well as queries referring to unavailable data (e.g., questions for delays of specific flights). All analyzed queries were relatively short (maximum length of 90 characters). As shown in Figure~\ref{fig:queryComplexity}, the analyzed queries restrict between zero and two dimension columns (we had one single voice query with three predicates during one of our user studies). Most queries fall into the category of retrieval queries (supported) while fewer queries ask for extrema or relative comparisons. We did not encounter any voice queries that translate into complex SQL queries using features such as exists condition or predicates on groups (i.e., having clauses). In summary, a relatively large portion of queries comply with the format considered in pre-processing. %In particular for the democratic primary set, incoming queries are often open-ended (e.g., ``Tell me about the primaries'', ``How is the democratic primaries going'', or ``I want to hear the latest polls for the democratic primary''), motivating data summarization. 

\begin{figure}
    \centering
    \begin{tikzpicture}
        \begin{groupplot}[group style={group size=2 by 1, horizontal sep=1.25cm}, width=3.5cm, height=3.25cm, xtick={1,2,3}, xticklabels={S, F, P}, xticklabel style={rotate=0, font=\small}, ybar=0pt, xmin=0.5, xmax=3.5, ymin=10, ymajorgrids, ymode=log, legend entries={This (Pre-Processing), This (Run Time), Baseline}, legend to name=processingLegend, legend columns=3]
            \nextgroupplot[ylabel={Latency (ms)}, bar width=5pt]
            \addplot coordinates {(1,74) (2,60) (3,54)};
            \addplot+[xshift=-5pt] coordinates {(1,74) (2,60) (3,54)};
            \addplot+[xshift=-5pt] coordinates {(1,500) (2,500) (3,500)};
            \nextgroupplot[ylabel={Time (ms)}, ybar=0pt, bar width=5pt]
            \addplot coordinates {(1,135) (2,85) (3,75)};
            \addplot+[xshift=-5pt] coordinates {(1,74) (2,60) (3,54)};
            \addplot+[xshift=-5pt] coordinates {(1,9769) (2,9833) (3,9000)};
        \end{groupplot}
    \end{tikzpicture}
    \ref{processingLegend}
    \caption{Average latency and per-query processing time for Stackoverflow (S), Flights (F), and Primaries (P) data sets.}
    \label{fig:runTimeOverheads}
\end{figure}
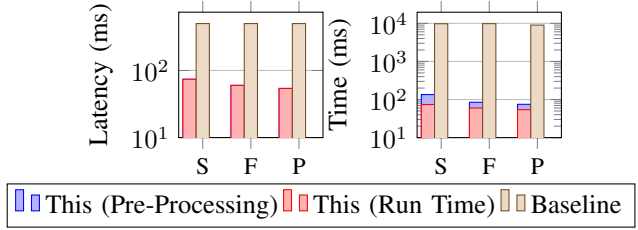

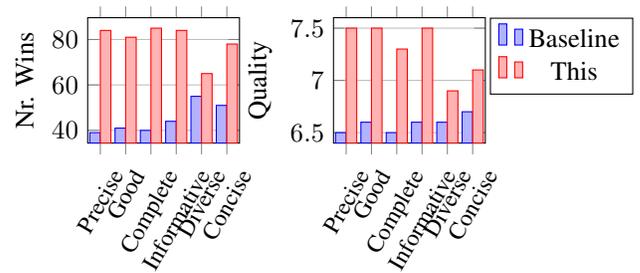
\begin{figure}[t]
\centering
\begin{tikzpicture}
\begin{groupplot}[group style={group size=2 by 1, horizontal sep=1.25cm}, width=3.6cm, height=3.25cm, xticklabels={Precise, Good, Complete, Informative, Diverse, Concise}, ymajorgrids, xticklabel style={rotate=60, font=\small}, xtick={1,2,3,4,5,6}, xmin=0.5, xmax=6.5, ylabel near ticks]
\nextgroupplot[ybar=0pt, bar width=4pt, ylabel={Nr.\ Wins}, legend entries={Baseline, This Approach}, legend to name=vsBaselineUserLegend]
\addplot coordinates {(1,39) (2,41) (3,40) (4,44) (5,55) (6,51)};
\addplot coordinates {(1,84) (2,81) (3,85) (4,84) (5,65) (6,78)};
\nextgroupplot[ybar=0pt, bar width=4pt, ylabel={Quality}, legend entries={Baseline, This}, legend columns=1, legend pos=outer north east]
\addplot coordinates {(1,6.5) (2,6.6) (3,6.5) (4,6.6) (5,6.6) (6,6.7)};
\addplot coordinates {(1,7.5) (2,7.5) (3,7.3) (4,7.5) (5,6.9) (6,7.1)};
\end{groupplot}
\end{tikzpicture}

%\ref{vsBaselineUserLegend}
\caption{Preferences of AMT workers when comparing different speeches describing the same data.\label{fig:vsBaselineUser}}
\end{figure}

%Maximally 90 characters per request. Longest request asks to compare salary between expert and novice coders. 

\subsection{Comparison to Baselines}
\label{sub:vsBaselines}

Next, we compare the greedy algorithm with optimized pruning against a recently proposed data vocalization method~\cite{Trummer2018a, Trummer2019a}. The prior method uses sampling to approximate the quality of a multitude of candidate speeches. Sampling creates non-negligible overheads at run time but can be partially overlapped with speech output. Hence, latency (i.e., time until the system starts speaking) is significantly lower than total sampling time. We report both, latency and processing time, in Figure~\ref{fig:runTimeOverheads}. We average all metrics over the queries analyzed in the last subsection (the supported data access queries among them). Our approach has very low run time overheads as it merely looks up the best pre-generated speech. For the baseline, latency is determined by the time required to select the first sentence via sampling. In exchange, we spend a total of 25 minutes in pre-processing to generate speech answers for a total of 28,720 queries, considering all three data sets. As long as data remain static, significant pre-processing overheads can be amortized over many queries. In Figure~\ref{fig:runTimeOverheads}, we also show per-query pre-processing overheads for each data set. The speeches generated by the baseline provide only value ranges as opposed to specific averages (to account for imprecision of sampling). Hence, we cannot calculate utility according to our utility model. However, we asked crowd workers to compare speech descriptions describing the same data, generated by the baseline and our approach. We use the same experimental setup as for the experiments reported in Figure~\ref{fig:userPref} (except that we add two more adjectives, ``Concise'' and ``Diverse''). We compare speeches for flights in general, for flights in the Northeast, and for flights in the Northeast and in Winter (three queries that were used for experiments by the prior publication~\cite{Trummer2018a}). Figure~\ref{fig:vsBaselineUser} shows results based on 900 HITs. Reporting precise values (e.g., ``the cancellation probability is 6\%'' as opposed to ``the cancellation probability is between 5 and 10\%'') likely leads to gains for properties like ``Precise'' and ``Informative''.

Finally, we performed a small experiment in which we trained a machine learning (ML) model with pairs of speech fragments and summaries (i.e., a text containing facts considered by our approach and the summary generated). Our goal was to verify whether ML methods can use a small seed set of summaries to generate the remaining ones. We focus on speeches for a common query template (i.e., for queries with predicates on the same combination of dimension columns) which tend to be similar. For the flights data set, we selected the dimension with the largest number of distinct values (start airport region with 52 values) and consider all queries placing one predicate on that dimension. Our implementation is based on the Simpletransformers library\footnote{https://simpletransformers.ai/}. Internally, it uses a sequence-to-sequence model based on pre-trained language models (an approach known to require orders of magnitude less training samples compared to prior methods~\cite{Howard2018}). We trained for 10 epochs on a Google CoLab standard GPU. We used 49 training samples (training took 30 seconds) and three samples for testing. ML predictions are fast (24 milliseconds per sample) and generate speeches that use similar syntactic patterns as ours. However, the ML-generated speeches are often redundant (multiple facts in the same speech referencing the same dimension) and tend to focus on overly narrow data subsets (e.g., cancellations in specific months instead of seasons). We created 900 AMT HITs to compare ML-generated speeches to the ones generated by our approach, according to the adjectives shown in Figure~\ref{fig:vsBaselineUser} and using the same experimental setup. The ML-generated speeches were consistently ranked lower (average rating below 5.92 for any adjective) compared to the proposed approach (average ratings of more than 7.28 for any adjective). ML-based summarization is showing promise but will likely require specialized variants for satifsactory performance.

\section{Related Work}
\label{relatedSec}

%Our work relates in particular to three research domains: vocalization, visualization, and approaches for efficient OLAP. 

%The focus in this work is on efficiently generating speech data summaries. In that, it relates to prior work on data vocalization

%Our work differs by the technique we exploit to reduce speech generation overheads. 

This work connects to prior work on generating speech summaries (``data vocalization'')~\cite{Trummer2017, Trummer2018, Trummer2018a, Trummer2019a}. While prior work exploits sampling or incremental processing, we generate optimal speeches in a pre-processing step. We compare against a corresponding baseline in Section~\ref{sub:vsBaselines}.

%This is motivated by a focus on Cloud scenarios where a backend serves multiple users. In that scenario, the goal is to reduce computational overheads (and therefore monetary execution fees). In contrast to that, prior work on incremental processing~\cite{Trummer2019a} has focused on minimizing latency perceived by the user (without reducing processing overheads). Prior work on sampling for vocalization still requires run time processing and often requires to overlap speech output with sampling to achieve near-optimal speech quality~\cite{Trummer2018a}. Again, this leads to elevated computation overheads at run time. 

Our work connects to work on voice and multi-modal interfaces~\cite{Lyons2016, Saul2019, Simitsis2008a, Yang2019}. It differs by our focus on summarizing large data sets via voice output. Prior work on natural language to SQL interfaces~\cite{Arbor2016, Li2014, Yahya2012, Francia2020, Affolter2019, Saha2016} addresses challenges in translating text to queries. Instead, our primary focus is on translating query results into concise text.

%Generally, an end-to-end voice interface requires techniques for speech recognition and natural language processing. So far, we use existing components while focusing on the output side. % (i.e., generation of voice descriptions). 

%Our query and speech model resemble the ones used for OLAP (since we consider aggregates, broken down by data subset). Hence, o

Our approach connects to prior work optimizing which visualizations or OLAP cubes to show to users~\cite{Joglekar2015a, Marcel2012, Sarawagi1998, Sarawagi2000, Sarawagi2000a, Vartak2014, Wongsuphasawat2015}. Typically, this work optimizes at the granularity of plots, characterized by breakdown attributes and potentially data subsets. The search space in speech generation is larger as each fact can refer to different attributes and different data subsets. At the same time, the amount of information that can be transmitted via speech tends to be lower, motivating higher efforts in picking the right pieces of information to transmit.   As opposed to non-speech sonifications~\cite{Hermann2011, Ramloll2001} (e.g., translating time series into a pitch), we focus on generating speech descriptions instead. % The main difference between our and prior work is the focus on voice output (as opposed to non-speech sounds) and large data sets. This motivates our approach for interleaved processing and voice output. 

We can state our problem as summarizing~\cite{Jing2000} a large text, containing various facts, concisely. We try a simple variant in Section~\ref{sub:vsBaselines}. Our work relates to prior work on data-to-text generation. Early approaches use domain-specific rules~\cite{Kittredge1994, McKeown1995}. Recent work~\cite{Puduppully2019} exploits machine learning to learn summarizing data for certain domains, based on training samples. Most prior work focuses on generating multi-paragraph summaries that would be too long for voice output. Our problem model, maximizing utility by a bounded number of facts, is motivated by the conciseness constraints particular to voice summaries. Our method uses standard SQL operations and can be executed without moving data out of the database system. Furthermore, our approach can be applied without training samples or hand-crafted summarization rules. 

%The system using our speech summarization methods was recently demonstrated at the VLDB conference~\cite{Trummer2020}.
%Using our method to generate training samples for a learning-based approach is an interesting avenue for future work.

%\cite{Puduppully2019}~trained on thousands of summaries

%\cite{Hassan2014}

\section{Conclusion}
\label{conclusionSec}
%presented an approach that prepares speech answers for voice queries in a pre-processing step, thereby reducing run time overheads. 
We generate speech answers to voice queries efficiently via pre-processing. Our approach reduces run time overheads and was validated in a public deployment.

%Voice query interfaces are motivated by scenarios in which visual interfaces are not available, due to constraints of devices (e.g., smart speakers), due to the user context (e.g., data access while driving), or due to other reasons (e.g., visual impairment). We presented an efficient approach that reduces run time overheads for voice querying to almost zero. We proposed and compared different pre-processing algorithms in terms of result quality and overheads. Furthermore, we validated the proposed approach in several user studies and in a public deployment.

%We calculate near-optimal speech summaries in a pre-processing step, leveraging several methods to increase efficiency. 

%Nevertheless, we have compared voice against visual query interfaces in several user studies. While visual interfaces enable users to explore data at a higher degree of detail, lay users with little experience in data analysis often consider voice interfaces a convenient alternative.

%Voice querying is a novel area and various research challenges remain unsolved. Among our future work plans are extensions to the supported query model as well as extensions to new types of data.

\bibliographystyle{abbrv}
\bibliography{library}

\begin{thebibliography}{10}

\bibitem{Affolter2019}
K.~Affolter, K.~Stockinger, and A.~Bernstein.
\newblock {A comparative survey of recent natural language interfaces for
  databases}.
\newblock {\em VLDB Journal}, 28(5):793--819, 2019.

\bibitem{Francia2020}
M.~Francia, E.~Gallinucci, and M.~Golfarelli.
\newblock {Towards conversational OLAP}.
\newblock {\em CEUR Workshop Proceedings}, 2572:6--15, 2020.

\bibitem{Kittredge1994}
E.~Goldberg, N.~Driedger, and R.~Kittredge.
\newblock {Using natural-language processing to produce weather forecasts}.
\newblock {\em IEEE Expert-Intelligent Systems and their Applications},
  9(2):45--53, 1994.

\bibitem{Hermann2011}
T.~Hermann, A.~Hunt, and J.~G. Neuhoff.
\newblock {\em {The Sonification Handbook}}.
\newblock 2011.

\bibitem{Howard2018}
J.~Howard and S.~Ruder.
\newblock {Universal Language Model Fine-tuning for Text Classification}.
\newblock In {\em ACL}, pages 328--339, 2018.

\bibitem{Jing2000}
H.~Jing and K.~R. McKeown.
\newblock {Cut and paste based text summarization}.
\newblock In {\em ACL}, pages 178--185, 2000.

\bibitem{Joglekar2015a}
M.~Joglekar, H.~Garcia-molina, and A.~Parameswaran.
\newblock {Smart drill down}.
\newblock {\em VLDBJ}, 8(12):1928--1931, 2015.

\bibitem{Krause2012}
A.~Krause and D.~Golovin.
\newblock {Submodular function maximization}.
\newblock Technical report, 2012.

\bibitem{Li2014}
F.~Li and H.~Jagadish.
\newblock {NaLIR: an interactive natural language interface for querying
  relational databases}.
\newblock {\em SIGMOD}, pages 709--712, 2014.

\bibitem{Arbor2016}
F.~Li and H.~Jagadish.
\newblock {Understanding natural language queries over relational databases}.
\newblock {\em SIGMOD Record}, 45(1):6--13, 2016.

\bibitem{Lyons2016}
G.~Lyons, V.~Tran, C.~Binnig, U.~Cetintemel, and T.~Kraska.
\newblock {Making the case for Query-by-Voice with EchoQuery}.
\newblock In {\em SIGMOD}, pages 2129--2132, 2016.

\bibitem{Marcel2012}
P.~Marcel, P.~J. Jaur{\`{e}}s, and S.~Rizzi.
\newblock {Towards intensional answers to OLAP queries for analytical
  sessions}.
\newblock In {\em DOLAP}, pages 49--56, 2012.

\bibitem{McKeown1995}
K.~McKeown, J.~Robin, and K.~Kukich.
\newblock {Generating concise natural language summaries}.
\newblock {\em Information Processing and Management}, 31(5):703--733, 1995.

\bibitem{MILLER1956}
G.~A. Miller.
\newblock {The magical number 7, plus or minus 2 - some limits on our capacity
  for processing information}.
\newblock {\em Psychological Review}, 63(2):81--97, 1956.

\bibitem{Nemhauser1978}
G.~Nemhauser and L.~Wolsey.
\newblock {Best algorithms for approximating the maximum of a submodular set
  function}.
\newblock {\em Mathematics of Operations Research}, 3(3):177--188, 1978.

\bibitem{Puduppully2019}
R.~Puduppully, L.~Dong, and M.~Lapata.
\newblock {Data-to-text generation with content selection and planning}.
\newblock In {\em AAAI}, pages 6908--6915, 2019.

\bibitem{Ramloll2001}
R.~Ramloll, W.~Yu, and B.~Riedel.
\newblock {Using non-speech sounds to improve access to 2D tabular numerical
  information for visually impaired users}.
\newblock In {\em Conference of the British HCI Group}, pages 515--529, 2001.

\bibitem{Saha2016}
D.~Saha, A.~Floratou, K.~Sankaranarayanan, U.~F. Minhas, A.~R. Mittal, and
  F.~Ozcan.
\newblock {ATHENA: An ontology-driven system for natural language querying over
  relational data stores}.
\newblock {\em VLDB}, 9(12):1209--1220, 2016.

\bibitem{Sarawagi2000}
S.~Sarawagi.
\newblock {User-adaptive exploration of multidimensional data}.
\newblock In {\em VLDB}, pages 307--316, 2000.

\bibitem{Sarawagi1998}
S.~Sarawagi, R.~Agrawal, N.~Megiddo, and S.~Jose.
\newblock {Discovery-driven exploration of OLAP data cubes}.
\newblock In {\em EDBT}, pages 168--182, 1998.

\bibitem{Sarawagi2000a}
S.~Sarawagi and G.~Sathe.
\newblock {i3: Intelligent, Interactive Investigaton of OLAP data cubes}.
\newblock {\em SIGMOD Record}, 29(2):589, 2000.

\bibitem{Saul2019}
V.~Shah, S.~Li, A.~Kumar, and L.~Saul.
\newblock {SpeakQL: towards speech-driven multimodal querying of structured
  data}.
\newblock Technical report, 2019.

\bibitem{Yang2019}
V.~Shah, S.~Li, K.~Yang, A.~Kumar, and L.~Saul.
\newblock {Demonstration of SpeakQL: speech-driven multimodal querying of
  structured data}.
\newblock In {\em SIGMOD Demo Track}, pages 2001--2004, 2019.

\bibitem{Simitsis2008a}
A.~Simitsis, Y.~Alexandrakis, G.~Koutrika, and Y.~Ioannidis.
\newblock {Synthesizing structured text from logical database subsets}.
\newblock In {\em EDBT}, pages 428--439, 2008.

\bibitem{Trummer2019a}
I.~Trummer.
\newblock {Data Vocalization with CiceroDB}.
\newblock In {\em CIDR}, 2019.

\bibitem{Trummer2020}
I.~Trummer.
\newblock {Demonstrating the voice-based exploration of large data sets with
  CiceroDB-zero}.
\newblock {\em Proceedings of the VLDB Endowment}, 13(12):2869--2872, 2020.

\bibitem{Trummer2018}
I.~Trummer, M.~Bryan, and R.~Narasimha.
\newblock {Vocalizing large time series efficiently}.
\newblock {\em PVLDB}, 11(11):1563--1575, 2018.

\bibitem{Trummer2018a}
I.~Trummer, Y.~Wang, and S.~Mahankali.
\newblock {A holistic approach for query evaluation and result vocalization in
  voice-based OLAP}.
\newblock In {\em SIGMOD}, pages 936--953, 2019.

\bibitem{Trummer2017}
I.~Trummer, J.~Zhu, and M.~Bryan.
\newblock {Data vocalization: optimizing voice output of relational data}.
\newblock {\em PVLDB}, 10(11):1574--1585, 2017.

\bibitem{Vartak2014}
M.~Vartak, S.~Madden, A.~Parameswaran, and N.~Polyzotis.
\newblock {SeeDB: automatically generating query visualizations}.
\newblock {\em VLDB}, 7(13):1581--1584, 2014.

\bibitem{Wongsuphasawat2015}
K.~Wongsuphasawat, D.~Moritz, A.~Anand, J.~Mackinlay, B.~Howe, and J.~Heer.
\newblock {Voyager: exploratory analysis via faceted browsing of visualization
  recommendations}.
\newblock {\em Transactions on Visual and Computer Graphics}, 22(1):649--658,
  2015.

\bibitem{Yahya2012}
M.~Yahya, K.~Berberich, S.~Elbassuoni, M.~Ramanath, V.~Tresp, and G.~Weikum.
\newblock {Natural language questions for the web of data}.
\newblock {\em EMNLP -CoNLL '12}, (July):379--390, 2012.

\end{thebibliography}


\begin{thebibliography}{00}
\bibitem{b1} G. Eason, B. Noble, and I. N. Sneddon, ``On certain integrals of Lipschitz-Hankel type involving products of Bessel functions,'' Phil. Trans. Roy. Soc. London, vol. A247, pp. 529--551, April 1955.
\bibitem{b2} J. Clerk Maxwell, A Treatise on Electricity and Magnetism, 3rd ed., vol. 2. Oxford: Clarendon, 1892, pp.68--73.
\bibitem{b3} I. S. Jacobs and C. P. Bean, ``Fine particles, thin films and exchange anisotropy,'' in Magnetism, vol. III, G. T. Rado and H. Suhl, Eds. New York: Academic, 1963, pp. 271--350.
\bibitem{b4} K. Elissa, ``Title of paper if known,'' unpublished.
\bibitem{b5} R. Nicole, ``Title of paper with only first word capitalized,'' J. Name Stand. Abbrev., in press.
\bibitem{b6} Y. Yorozu, M. Hirano, K. Oka, and Y. Tagawa, ``Electron spectroscopy studies on magneto-optical media and plastic substrate interface,'' IEEE Transl. J. Magn. Japan, vol. 2, pp. 740--741, August 1987 [Digests 9th Annual Conf. Magnetics Japan, p. 301, 1982].
\bibitem{b7} M. Young, The Technical Writer's Handbook. Mill Valley, CA: University Science, 1989.
\end{thebibliography}

\begin{comment}

\end{comment}

\end{document}